\documentclass[11pt,reqno]{amsart}
\usepackage{amsmath,amsthm}
\usepackage{latexsym}
\usepackage{amsfonts}
\usepackage{amssymb}
\usepackage{color}
\usepackage{bbm,dsfont}
\usepackage{graphicx}
\usepackage{textcomp}
\usepackage[pagebackref, colorlinks = true, linkcolor = blue, urlcolor  = blue, citecolor = red]{hyperref}
\usepackage[margin=1in]{geometry}
\usepackage{array}
\usepackage[usenames,dvipsnames]{xcolor}

\newtheorem{theorem}{Theorem}[section]
\newtheorem{proposition}[theorem]{Proposition}
\newtheorem{lemma}[theorem]{Lemma}
\newtheorem{corollary}[theorem]{Corollary}
\newtheorem{definition}[theorem]{Definition}

\newtheorem{remark}[theorem]{Remark}
\newtheorem{example}[theorem]{Example}

\newcommand{\nat}{\mathbb N} %natural
\newcommand{\half}{\tfrac{1}{2}} %half
\newcommand{\hi}{\mathcal{H}} %Hilbert space
\newcommand{\lh}{\mathcal{L(H)}} %bounded linear operators
\newcommand{\ket}[1]{|#1\rangle} %ket
\newcommand{\bra}[1]{\langle#1|} %bra
\newcommand{\kb}[2]{|#1\rangle\langle#2|} %ketbra
\newcommand{\no}[1]{\left\|#1\right\|} %norm
\newcommand{\tr}[1]{\mathrm{Tr}\left[#1\right]} %trace
\newcommand{\id}{\mathbbm{1}} %identity operator
\newcommand{\Mob}{\operatorname{M\ddot{o}b}}
\newcommand{\Wg}{\operatorname{Wg}}
\renewcommand{\id}{\mathrm{id}}

\begin{document}

\title{Random positive operator valued measures}

\author{Teiko Heinosaari}

\address{TH: QTF Centre of Excellence, Turku Centre for Quantum Physics, Department of Physics and Astronomy, University of Turku, Turku 20014, Finland}
\email{teiko.heinosaari@utu.fi}

\author{Maria Anastasia Jivulescu}
\address{MAJ: Department of Mathematics,
Politehnica University of Timi\c soara,
Victoriei Square 2, 300006 Timi\c soara, Romania}
\email{maria.jivulescu@upt.ro}

\author{Ion Nechita}
\address{IN: Laboratoire de Physique Th\'eorique, Universit\'e de Toulouse, CNRS, UPS, France}
\email{nechita@irsamc.ups-tlse.fr}

\begin{abstract}
We introduce several notions of random positive operator valued measures (POVMs), and we prove that some of them are equivalent. We then study statistical properties of the effect operators for the canonical examples, starting from the limiting eigenvalue distribution. We derive the large system limit for several quantities of interest in quantum information theory, such as the sharpness, the noise content, and the probability range. Finally, we study different compatibility criteria, and we compare them for generic POVMs. 
\end{abstract}

\date{\today}

\maketitle

\tableofcontents

%%%%%%%%%%%%%%%%%%%%%%%
\section{Introduction}\label{sec:intro}
%%%%%%%%%%%%%%%%%%%%%%%

In the last few years, significant developments have been reported in Quantum Information Theory as a consequence of applying sophisticated techniques coming from  Random Matrix Theory and Free Probability Theory. Indeed, the introduction of suitable models for random quantum states and channels has generated results in various topics, such as: quantum entanglement \cite{aubrun2014entanglement}, classical capacity of quantum channels \cite{fukuda2018minimum}, additivity question \cite{hastings2009superadditivity,collins2016random}. It is of interest to apply such methods to other concepts or open problems from quantum information, such as, for example, \textit{positive operator-valued measures} (POVMs) \cite{heinosaari2012mathematical}.

In this paper we define \textit{random POVMs} and study thoroughly their properties. Moreover, we ask questions about  
the  (in)-compatibility of two independent random POVMs and find suitable conditions using various criteria from the literature.
 We actually present  several models of randomness for POVMs and we also study the connections between them. The most natural way to define a random POVM  is as the image of diagonal unit rank projections  through random unital, completely positive maps coming from Haar isometries.  This is the model that we shall consider mostly in this paper: 
\begin{definition}
	Fix an orthonormal basis $\{e_i\}_{i=1}^k$ of $\mathbb C^k$ and consider a Haar-distributed random isometry $V:\mathbb C^d \to \mathbb C^k \otimes \mathbb C^n$, for some integers $d,k,n$ with $d \leq kn$. Define the random unital, completely positive map $\Phi(X) =  V^*(X \otimes I_n)V$. A \emph{Haar-random POVM} is the $k$-tuple $(M_1,\ldots, M_k)$ defined by $M_i := \Phi(\ket{e_i} \bra{e_i})$.
\end{definition}
This model for random completely positive maps has been used also in other frameworks, see \cite[Section VI]{collins2016random} for a review. Two other models of randomness for POVMs are introduced: one coming from the Lebesgue measure on the compact set of POVMs and the other given by Wishart-random POVM ensemble. It is relevant to stress that in Theorem \ref{equiv} we prove the equivalence of the  Wishart-random POVMs model to the one coming from the Haar ensemble, while in Corollary \ref{cor:Lebesgue} we show that the Lebesgue measure is the special case $n=d$ in the definition above; these facts justify our choice in studying its  properties. We would like to mention that random POVMs have been previously considered in the literature: Naimark dilations to a random orthonormal basis of $\mathbb C^n \otimes \mathbb C^k$ were considered in \cite{radhakrishnan2009random}; in \cite{petz2012optimal}, the authors study Gaussian perturbations of a fixed POVM; normalized unit rank projections on i.i.d.~random vectors were considered in \cite{aubrun2016zonoids}, in a situation where the number of outcomes is larger than the dimension. Finally, in the work \cite{zhang2019incompatibility} (which appeared after the preprint version of our work was made available online), the authors compute several probabilities for the compatibility of independent dichotomic qubit POVMs, parametrized by points on the Bloch sphere. 

Using the most general Wishart model, we analyze the spectral distribution of the effect operators, which are elements of the \emph{Jacobi ensemble} \cite{wachter1980limiting}. We compute in Proposition \ref{prop-mom} the moments of the individual effects from a Haar-random POVM using (graphical) Weingarten calculus. In Proposition \ref{prop:limit-eigenvalues-Mi}, we re-derive the asymptotic spectral distribution of random effect, as a dilatation of free additive convolution of a Bernoulli measure. These results are of help for deriving auxiliary properties of random POVMs which involve spectral expressions, such as regularity, the norm-1 property, or the probability range. Furthermore, we study and compare (in)-compatibility criteria for Haar-random POVMs, such as the noise content criterion, the Jordan product criterion, the optimal cloning criterion, the Miyadera-Imai criterion, and the Zhu criterion.  Our study shows that, for certifying compatibility for typical random POVMs, it is of interest to check first the Jordan product criterion.

The paper is organized as follows. In Section \ref{sec:POVMs} we recall basic notions related to POVMs, definitions, relevant examples and remarks. Section \ref{sec:compatibility} deals the notion of compatible POVMs and contains a brief presentation of the known incompatibility criteria. In Section \ref{sec:interlude} we review the basic ingredients needed for a good understanding of random matrix theory techniques used in the paper. To this aim, different topics are approached, such as random isometries, Weingarten calculus (also in its graphical incarnation), as well as some tools from Voiculescu's free probability theory. In Section \ref{Secc:random-POVM} we describe in details the models of randomness for POVMs and we state remarks related to their equivalence. We present in Section \ref{sec:stat-prop} the statistical properties of random POVMs, whereas in Section \ref{Sec:comp-criteria} we consider incompatibility criteria for them, which are compared in Subsection \ref{Sec:comp}.

Before we move on, let us introduce some basic notation. 
We write $[n]:=\{1,2,\ldots, n\}$ and we denote by $\mathcal S_n$ the symmetric group acting on $[n]$. For a given permutation $\sigma\in\mathcal{S}_n$, we use the following notations: $\# \sigma$ is the number of cycles of $\sigma$ and $|\sigma|$ is the length of $\sigma$, that is the minimal number of transpositions that multiply to $\sigma$. We denote by $\gamma:=(n,\ldots,3,2,1)\in\mathcal{S}_n$ the full cycle permutation (in reverse order).
In this paper the following asymptotic notation is  used:
$$ x_n\sim y_n\Leftrightarrow \lim_{n\rightarrow \infty}\frac{x_n}{y_n}=1.$$
We denote by $\hi_d$ a finite $d$-dimensional complex Hilbert space and by $\mathcal{L}(\hi_d)$ the algebra of linear operators on $\hi_d$. 
Further, we denote by $\operatorname{Tr}$ the (un-normalized) trace of matrices. 

\bigskip
\noindent \textit{Acknowledgments.} I.N.'s research has been supported by the ANR projects {StoQ} {ANR-14-CE25-0003-01} and {NEXT} {ANR-10-LABX-0037-NEXT}. I.N.~acknowledges the hospitality of the universities  Politehnica Timi{\c s}oara and Turku. T.H.~acknowledges financial support from the Academy of Finland via the Centre of Excellence program (Project no.
312058) as well as Project no. 287750.

%%%%%%%%%%%%%%%%%
\section{POVMs and their properties}
\label{sec:POVMs}
%%%%%%%%%%%%%%%%%

The states of a quantum system are mathematically described as density operators on a complex Hilbert space $\hi_d \cong \mathbb C^d$, i.e.~positive semi-definite operators with unit trace. 
A measurement is, mathematically speaking, a map that assigns a probability distribution to every state. 
The probability distribution is interpreted as the distribution of measurement outcomes. 
The additional requirement is that this kind of map is affine; a convex mixture of states must go into the respective mixture of the probability distributions. 
It follows that quantum measurements can be identified with positive operator valued measures (POVMs) \cite{heinosaari2012mathematical}. 
We will only consider POVMs with finite number of outcomes and Hilbert spaces are assumed to be finite dimensional. 
In this section we recall some physically motivated properties of POVMs.

%%%%%%%%%%%%%%%%%
\subsection{POVMs}
%%%%%%%%%%%%%%%%%

For a POVM $A$, we denote by $\Omega_A$ the set of all outcomes of $A$, $\Omega_A=\{1,\ldots,k\}$ for some $k\in\nat$.
A POVM is then a map 
$$
A:\Omega_A\to\lh \, , \quad i \mapsto A_i 
$$
such that $\sum_i A_i = I$ ($=$ the identity operator on $\hi_d$) and $A_i$ are positive semi-definite operators, $A_i \geq 0$ for all $i \in \Omega_A$. The operators $A_i$ are called the \emph{effects} of the POVM $A$.

\begin{example}\label{ex:trivial}
A POVM $T$ is called \emph{trivial} if $T_i$ is proportional to the identity operator $I$ for every outcome $i \in \Omega_T$. 
In this case, there is a probability distribution $p$ on $\Omega_T$ such that $T_i = p_i I$.
\end{example}

\begin{example}\label{ex:sharp}
Let $\{\varphi_i\}_{i=1}^d$ be an orthonormal basis of $\hi_d$.
We set $A_i = \kb{\varphi_i}{\varphi_i}$ for every $i=1,\ldots,d$, and then $A$ is a POVM.
It is called the \emph{POVM associated to the orthonormal basis $\{\varphi_i\}_{i=1}^d$}.
\end{example}

Both types of POVMs from the previous examples are \emph{commutative}, i.e.~$A_iA_j=A_jA_i$ for all $i,j\in\Omega_A$.
One can easily construct examples of non-commutative POVMs by mixing two  POVMs corresponding to two different orthonormal bases. 
There are also non-commutative POVMs that are extreme in the set of all POVMs with the same outcome set; we refer to \cite{perinotti05} and 
\cite{haapasalo12} for further examples.

%%%%%%%%%%%%%%%%%
\subsection{Operator range and probability range of a POVM}
%%%%%%%%%%%%%%%%%

Let us first observe that when we have a measurement device that implements a POVM $A$, we can obtain not only the numbers $\operatorname{Tr}(\rho A_j)$, but also all sums of these numbers simply by grouping the measurement outcomes differently. 
For this reason, the following concept is useful when we talk about properties of POVMs.

\begin{definition}
For a POVM $A$ and a subset $X\subseteq\Omega_A$, we denote $A_X := \sum_{i \in X} A_i$.
The \emph{(operator) range} of $A$ is the set 
$$
\operatorname{Ran}(A):=\{ A_X: X\subseteq\Omega_A \}.
$$
\end{definition}

Instead of starting from POVMs, one can consider a measurement as an affine map from the state space to a probability simplex 
$$
\Delta_k=\left\{(p_1,\ldots,p_k): \textrm{$0\leq p_i \leq 1$ for all $1 \leq i \leq k$, $\sum_i p_i =1$}\right\}.
$$
It is well known that these descriptions are equivalent; any POVM determines such an affine map, and any such affine map determines a unique POVM \cite{holevo1973statistical}.
From this point of view, it is of equal importance and interest to study both the range of the affine map related to a POVM and its operator range, see also \cite{BuCaLa95}.

\begin{definition}
	The \emph{probability range} of a $k$-outcome POVM $A$ is the convex subset of the probability simplex
	$$\operatorname{ProbRan}(A):=\{ (\operatorname{Tr}(\rho A_1),\ldots,\operatorname{Tr}(\rho A_k)) \, : \, \rho \text{ is a density matrix}\} \subseteq \Delta_k.$$
\end{definition}

A trivial POVM $T$, given as $T_i=p_i I$ for a probability distribution $t$, has probability range reduced to the single point $p=(p_1,\ldots,p_k)$. 
On the other side of the spectrum, it is easy to see that a POVM $A$ has full probability range, that is $\operatorname{ProbRan}(A) = \Delta_k$, if and only if its effects have all unit operator norm, $\|A_i\|=1$, for all $1 \leq i \leq k$, see Definition \ref{def:norm1} below. 

We present next two examples of probability ranges. First, let us consider the case of diagonal effects. Let us assume that the POVM effects $A_i$ are diagonal, $A_i = \operatorname{diag}(a_i)$, for some vectors $a_i \in [0,1]^d$ satisfying 
$$\forall j \in [d], \qquad \sum_{i=1}^k a_i(j) = 1.$$
Considering the vectors $\alpha_j \in [0,1]^k$, for $j \in [d]$, defined by $\alpha_j(i) = a_i(j)$, we have the following result.
\begin{proposition}
The probability range of a diagonal POVM $A$ is the polytope $\operatorname{conv}\{\alpha_1, \ldots, \alpha_d\}.$
\end{proposition}
\begin{proof}
First, note that the normalization condition for the POVM $A$ translates to the fact that the $\alpha_j$ are probability vectors. Next, for a unit vector $x \in \mathbb C^d$, we have
$$[\langle x, A_i x \rangle]_{i=1}^k = \left[\sum_{j=1}^d |x_j|^2 a_i(j)\right]_{i=1}^k = \sum_{j=1}^d |x_j|^2 \alpha_j,$$
proving the claim. 
\end{proof}
As an example, see the left panel in Figure \ref{fig:example-ProbRan} where we have depicted the probability range of the following diagonal 3-outcome POVM:
\begin{align}
\nonumber A_1 &= \mathrm{diag}\left(\frac 1 2, \frac 1 3, \frac 1 6,\frac 1 6 ,\frac 1 3, \frac 1 2\right)\\
\label{eq:diagonal-POVM} A_2 &= \mathrm{diag}\left(\frac 1 3, \frac 1 2, \frac 1 2,\frac 1 3 ,\frac 1 6, \frac 1 6\right)\\
\nonumber A_3 &= \mathrm{diag}\left(\frac 1 6, \frac 1 6, \frac 1 3,\frac 1 2 ,\frac 1 2, \frac 1 3\right).
\end{align}

The following example of a non-trivial probability range is taken from \cite{fukuda2015quantum}. 
Consider a 3-outcome qubit POVM, with unit rank effects 
$A_i = \frac 2 3 \ket{a_i}\bra{a_i}$, where 
\begin{equation}\label{eq:circle-POVM}
a_1 = \begin{bmatrix} 1 \\ 0 \end{bmatrix} \, , \qquad a_2 = \begin{bmatrix} -1/2 \\ \sqrt{3}/2 \end{bmatrix} \, , \qquad a_3 = \begin{bmatrix} -1/2 \\ -\sqrt{3}/2 \end{bmatrix}.
\end{equation}	
A direct computation shows that the squared distance from a point $(\operatorname{Tr}(\rho A_i))_{i=1}^3$ to the ``center'' $(1/3,1/3,1/3)$ of the probability simplex $\Delta_3$ is less than $[(1-2a)^2+4|b|^2]/6$, where $\rho$ is an arbitrary qubit density matrix
$$\rho = \begin{bmatrix} a & b\\\bar b & 1-a\end{bmatrix}.$$
Using the positivity condition for $\rho$, i.e.~$|b|^2\leq a(1-a)$, we conclude that the probability range of the POVM $A$ is contained in a circle of radius $1/\sqrt 6$ around the equiprobability vector $(1/3,1/3,1/3)$; doing the computations backwards shows that in fact we have equality between the probability range and the aforementioned circle, see Figure \ref{fig:example-ProbRan}, right panel.

\begin{figure}[htbp]
	\begin{center}
	\qquad 	\includegraphics[scale=0.60]{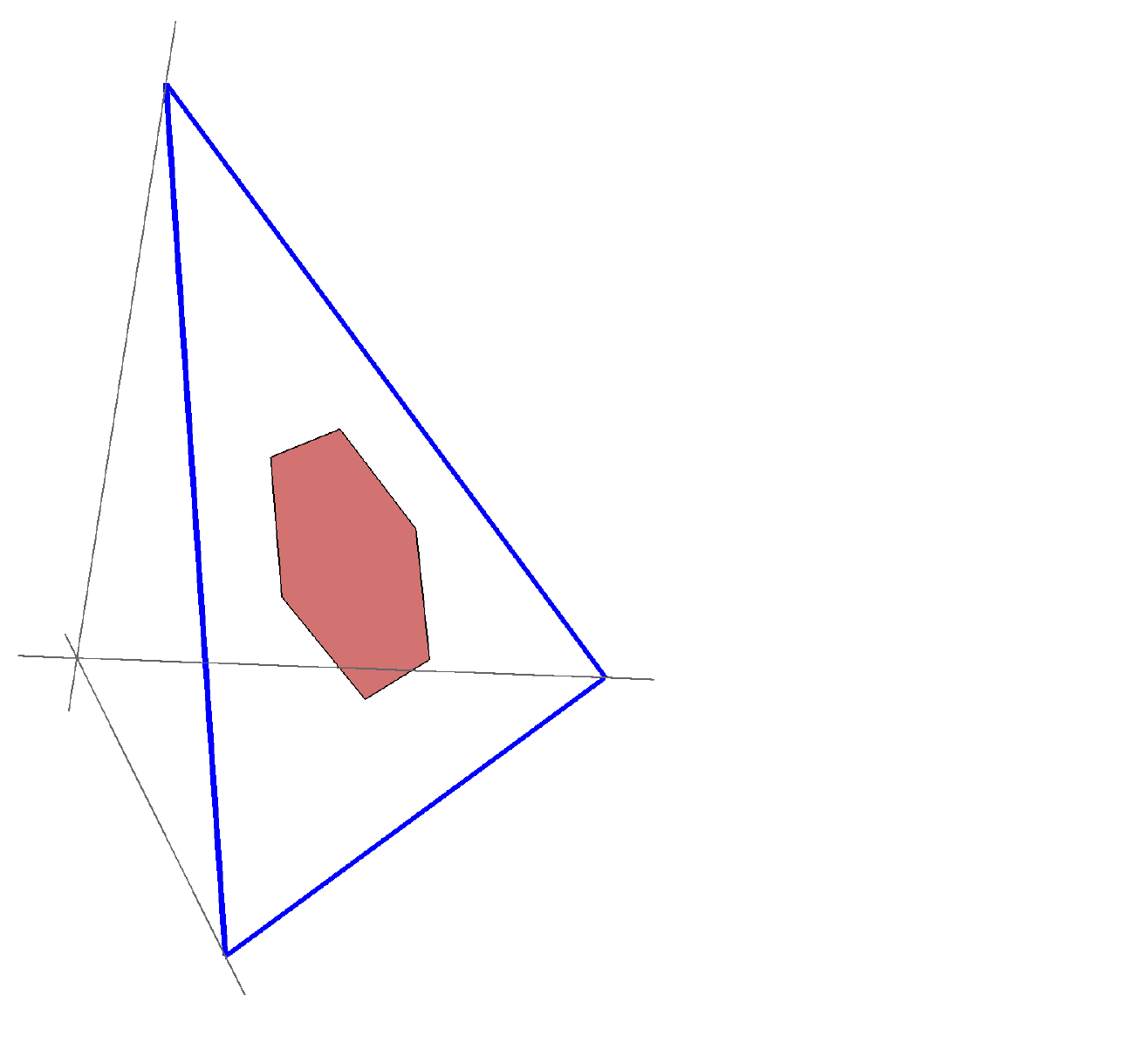}	\!\!\!\!\!\!\!\!\!\!\!\!\!\!\!	\includegraphics[scale=0.60]{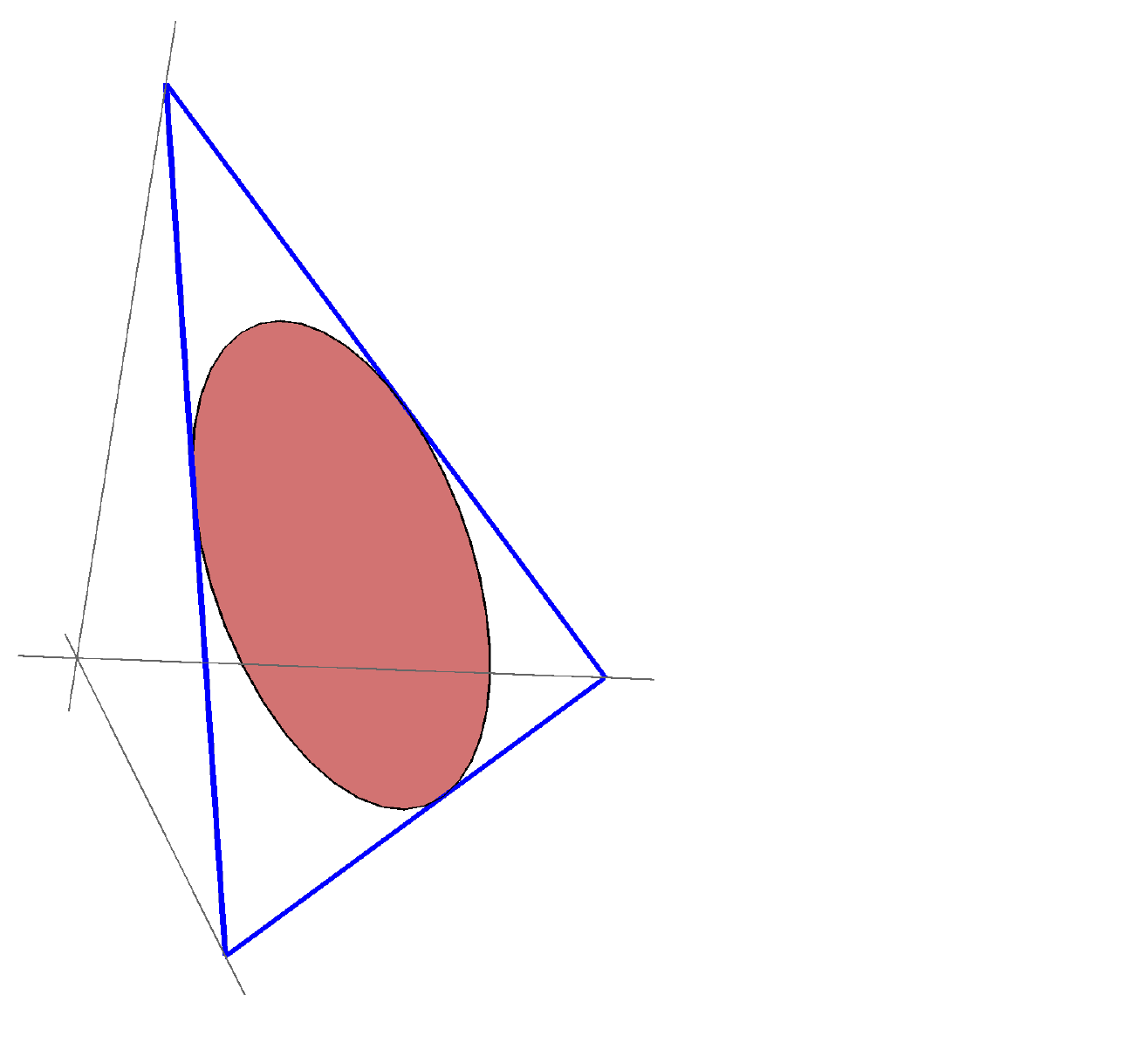}
		\caption{Examples for the probability range of two POVMs with 3 outcomes. On the left, the diagonal POVM from \eqref{eq:diagonal-POVM}. On the right, the example from \eqref{eq:circle-POVM}; the probability range is a disk around the equiprobability vector $(1/3,1/3,1/3)$. The axes are in gray, the probability simplex $\Delta_3$ is the blue triangle, and the probability range is the red convex set.}
		\label{fig:example-ProbRan}
	\end{center}
\end{figure}

%%%%%%%%%%%%%%%%%
\subsection{Spectral properties of POVMs}\label{sec:spectral-properties-POVMs}
%%%%%%%%%%%%%%%%%

This section contains a list of properties of quantum effects and POVMs relevant from the point of view of quantum information theory. 
We shall introduce them via a list of definitions followed by some simple properties and remarks; in the following sections, we shall study these properties for random POVMs. 
All these properties reduce to some property on the spectrum of the effects. 
We denote by $\operatorname{spec}(E)$ the spectrum of an operator $E$.

%%%%%%%%%%%%%%%%%
\subsubsection{Sharpness and regularity}
%%%%%%%%%%%%%%%%%

An effect $E$ is called \emph{sharp} if it is a projection (i.e.~$E^2=E$), and otherwise \emph{unsharp}. 
Hence, being sharp is equivalent to $\operatorname{spec}(E) \subseteq \{0,1\}$.
A POVM $A$ is called sharp if $A_i$ is sharp for every $i \in \Omega_A$; otherwise $A$ is called unsharp. As a measure of unsharpness, we use the following.

\begin{definition}\label{def:sharpness}
The \emph{unsharpness} of an effect $E$ is
\begin{equation}\label{eq:def-sharpness}
\sigma(E) := 4\no{E - E^2}.
\end{equation}
The unsharpness of a POVM $A$ is 
\begin{equation}\label{eq:def-sharpness-povm}
\sigma(A) := \max_i \sigma(A_i) \, .
\end{equation}
\end{definition}

For quantum effects, we have $0 \leq \sigma(E) \leq 1$, with $\sigma(E)=0$ iff $E$ is a sharp and $\sigma(E)=1$ iff $\tfrac{1}{2}\in\operatorname{spec}(E)$. For POVMs, it also holds that $0 \leq \sigma(A) \leq 1$. 

One may ask if there is a qualitative property between sharpness and unsharpness. This kind of property is regularity \cite{DvPu94,LaPu97}.

\begin{definition}\label{def:regular}
An effect $E$ is called \emph{regular} if neither $E \leq \half I$ nor $\half I\leq E$. 
A POVM $A$ is called regular if all effects, except $0$ and $I$, in $\operatorname{Ran}(A)$ are regular.
\end{definition}

For effects, the definition above is equivalent to the fact that the spectrum $\operatorname{spec}(E)$ of $E$ is not contained in $[0,\half]$ or $[\half,1]$. 
Interestingly, it can be shown \cite{DvPu94} that a POVM $A$ is regular if and only if $\operatorname{Ran}(A)$ is a Boolean lattice with respect to the operator order $\leq$ and the complementation $E \mapsto I-E$ restricted to $\operatorname{Ran}(A)$.

%%%%%%%%%
\subsubsection{Norm-1 property}
%%%%%%%%%

We recall the following definition \cite{HeLaPePuYl03}.

\begin{definition}\label{def:norm1}
A POVM $A$ has the \emph{norm-1-property} if $\no{A_i}=1$ for every $i$.
\end{definition}

Physically, the norm-1-property means that for each outcome $i$, there is a state $\rho_i$ such that the outcome $i$ occurs with certainty, i.e., $\tr{\rho_i A_i}=1$.
It follows that $\tr{\rho_i A_j}=0$ for $i \neq j$, thereby each operator $A_i$ has both the eigenvalues $0$ and $1$.
In particular, a POVM with the norm-1-property is regular. 

We recall another characterization of the norm-1-property that links to a different physical property.
An \emph{instrument} $\mathcal I$ is a mapping from an outcome set (here, a finite set) to the set of quantum operations (completely positive maps), satisfying the obvious normalization and additivity properties \cite[Section 5.1.2]{heinosaari2012mathematical}. Instruments encode the transformations of a quantum state following a measurement, so they contain more information than POVMs, which only deal with the probabilities of obtaining different outcomes. 
For any given POVM $A$, there are several instruments that describe some state transformation associated to some measurement of $A$ \cite{heinosaari2012mathematical}. 
An instrument $\mathcal{I}$ is called \emph{repeatable} if a subsequent measurement with the same device gives the same outcome, i.e., 
\begin{equation*}
\tr{\mathcal{I}_i ( \mathcal{I}_j(\rho))} = \delta_{ij} \tr{\mathcal{I}_i(\rho)} \, .
\end{equation*}	
As shown in \cite[Section III.4.6]{busch1996quantum}, a POVM $A$ admits a repeatable instrument if and only if $A$ has norm-1-property. The POVMs with the norm-1-property have also appeared in relation to a strong notion of additivity for quantum channels, see \cite[Definition 1 and Theorem 4]{fukuda2015quantum}.

%%%%%%%%%%%%%%
\subsubsection{Noise content}
%%%%%%%%%%%%%%

Trivial POVMs (see Example \ref{ex:trivial}) can be use to describe measurement noise.
Namely, if we start from a POVM $A$ and mix it with a trivial POVM $T$, then we get a noisy version of $A$.
Reversely, we can investigate how much noise a given POVM has. 
We recall the following definition \cite{filippov2017necessary}.

\begin{definition}
The \emph{noise content} $w(A)$ of a POVM $A$ is defined as
\begin{align*}
w(A):= \sup \{ 0 \leq t \leq 1 \, : \,  &A = t T + (1-t) B \text{ for some trivial POVM $T$}\\
& \quad \quad\text{and some POVM $B$ with $\Omega_T=\Omega_B=\Omega_A$} \} \, .
\end{align*}
\end{definition}

It can be shown \cite{filippov2017necessary} that
\begin{equation*}
w(A) = \sum_i \lambda_{min}(A_i) \, ,
\end{equation*}
where $\lambda_{min}(A_i)$ denotes the minimal eigenvalue of an operator $A_i$.

Instead of considering all trivial POVMs as noise, it is sometimes of interest to take only uniformly distributed trivial POVM as noise.
The \emph{uniform noise content} $w^u(A)$ of a POVM $A$ with $k$ outcomes is defined as
\begin{equation*}
w^u(A):= \sup \{ 0 \leq t \leq 1 : A = t \tfrac{1}{k}I + (1-t) B \quad \textrm{for some POVM $B$ with $\Omega_B=\Omega_A$} \} \, .
\end{equation*}
In this case, we define similarly $w^u(A) = \min_i \lambda_{min}(A_i)$. We note that the uniform noise content behaves very differently than the noise content. 
For instance, $w(T)=1$ for all trivial observables, whereas $w^u(T)=0$ for $T=p I$ such that $p_i=0$ for some outcome $i$.

%%%%%%%%%%%%%%%%%%%%%%%%%%%%%%
\section{Incompatibility of POVMs}
%%%%%%%%%%%%%%%%%%%%%%%%%%%%%%
\label{sec:compatibility}
%%%%

Mathematically, incompatibility is a $n$-place relation in the set of $n$-tuples of POVMs. In this work we concentrate only on the binary incompatibility relation. 
Physically speaking, incompatibility relation describes the impossibility to measure simultaneously two (or more) POVMs.
The simplest physical example of incompatible measurements consists of two different spin component measurements \cite{busch86}. 
The realm and applications of incompatibility have been extensively developed in the past years. We refer to \cite{heinosaari2016invitation} for a more extensive explanation and for further references.
In this section we recall all results on incompatibility that are needed later.

%%%%%%%%%%%%%%%%%%%%%%%%%%%%%%
\subsection{Definition and basic properties}
%%%%%%%%%%%%%%%%%%%%%%%%%%%%%%

Given two POVMs $A$ and $B$, we say that $B$ is a \emph{post-processing} of $A$ if there exists a column stochastic matrix $\mu$ such that
\begin{equation*}
B(x) = \sum_{y \in \Omega_A} \mu_{xy} A(y)
\end{equation*}
for all $x\in\Omega_B$.
The post-processing relation is a preorder on the set of POVMs and has been introduced in \cite{martens90}.

Two POVMs $A$ and $B$ are \emph{compatible} if there exists a third POVM $C$ such that $A$ and $B$ are both post-processings of $C$; otherwise $A$ and $B$ are \emph{incompatible}.
The compatibility relation is clearly reflexive and symmetric, but not transitive \cite{HeReSt08}.

We recall that if $A$ and $B$ are compatible, then they are marginals of a third POVM \cite{AlCaHeTo09}, called their joint POVM.
Namely, let us assume that $A(x)= \sum \mu^A_{xz} C(z)$ and $B(y)= \sum \mu^B_{yz} C(z)$ for some POVM $C$ and column stochastic matrices $A$ and $B$.
We then define a new POVM $G$ as
\begin{equation*}
G(x,y) = \sum_z \mu^A_{xz}\mu^B_{yz} C(z) \, .
\end{equation*}
Then
\begin{equation}\label{eq:marginals}
\sum_y G(x,y) = A(x) \, , \quad \sum_x G(x,y) = B(y) \, .
\end{equation}

We see that if $A$ and $B$ are compatible and $G$ satisfies \eqref{eq:marginals}, then $\operatorname{Ran}(A)\cup\operatorname{Ran}(B) \subset \operatorname{Ran}(G)$.
However, the existence of a POVM $G$ such that $\operatorname{Ran}(A)\cup\operatorname{Ran}(B) \subset \operatorname{Ran}(G)$ does not guarantee the compatibility of $A$ and $B$ \cite{ReReWo13}.

%%%%%%%%%%%%%%%%%%%%%%%%%%%%%%
\subsection{Criteria for compatibility}
%%%%%%%%%%%%%%%%%%%%%%%%%%%%%%

In the following we recall three sufficient conditions for compatibility or, in other words, necessary conditions for incompatibility. 
We present their proofs for the reader's convenience.

\begin{proposition}\label{prop:nc-criterion}
\emph{Noise content criterion} \cite{filippov2017necessary}: if two POVMs $A$ and $B$ satisfy
\begin{equation}\label{eq:nc-criterion}
w(A)+w(B) \geq 1\, ,
\end{equation}
then they are compatible.
\end{proposition}

\begin{proof}
If \eqref{eq:nc-criterion} holds, then there exist trivial POVMs $S_i=p_i I$ and $T_j=q_jI$ and numbers $s,t\in [0,1]$ such that $s+t = 1$ and
$A=sS + (1-s)A'$ and $B=tT + (1-t)B'$ for some POVMs $A',B'$.
We define a map $M$ as $M_{ij}:=t q_jA'_i+sp_i B'_j$.
Then $M$ is a joint POVM for $A$ and $B$.
\end{proof}

\begin{proposition}\label{prop:jordan-product-criterion}
\emph{Jordan product criterion} \cite{Heinosaari13}: if two POVMs $A$ and $B$ are such that
\begin{equation}\label{eq:jordan-product-criterion}
\forall i,j: \qquad A_i \circ B_j := A_i B_j + B_j A_i \geq 0,
\end{equation}
then $A$ and $B$ are compatible.
\end{proposition}

\begin{proof}
We define $M_{ij}= \half A_i \circ B_j$.
Then $\sum_j M_{ij} = A_i$ and $\sum_i M_{ij} = B_j$.
The requirement $A_i \circ B_j \geq 0$ implies that $M$ is a valid POVM.
\end{proof}

This Jordan product criterion covers as a special case the following well-known implication: if $A$ and $B$ commute, then they are compatible.

\begin{proposition}\label{prop:optimal-cloning}
\emph{Optimal cloning criterion} \cite{heinosaari2014maximally}: if two POVMs $A$ and $B$ satisfy  
\begin{align}
\label{eq:cloning-A}\forall i : {\lambda}_{min}(A_i) & \geq \frac{1}{2(1+d)} \tr{A_i} \, ,  \\
\label{eq:cloning-B}\forall j : {\lambda}_{min}(B_j) & \geq \frac{1}{2(1+d)} \tr{B_j} \, , 
\end{align}
then they are compatible.
\end{proposition}

\begin{proof}
We recall that the so-called symmetric universal cloning machine $\Lambda$, presented in \cite{keyl1999optimal}, is defined as
\begin{equation*}
\Lambda(\rho) = s_{d} \, S(\rho \otimes I) S\, , 
\end{equation*}
where $S$ is the projection from $\mathcal{H}_d^{\otimes 2}$ to the symmetric subspace of $\mathcal{H}_d^{\otimes 2}$ and the normalization coefficient $s_{d}$ is independent of $\rho$.
The state $\tilde{\rho}$ of each approximate copy is obtained as the corresponding marginal of $\Lambda(\rho)$ and, as it was shown in \cite{werner1998optimal}, it reads
\begin{equation*}
\tilde{\rho} = c_{d} \rho + (1-c_{d}) \tfrac{I}{d},
\end{equation*}
where the number $c_{d}$ is independent of $\rho$ and given by $c_{d} = (2+d)/(2+2d)$. 
We are then performing measurements of POVMs $A'$ and $B'$ on the two copies of $\rho$ obtained through the cloning machine. A measurement of $A'$ on the approximate copy $\tilde{\rho}$ gives the same result as the action of the noisy POVM $c_{d} A' + (1-c_{d}) T_{A'}$ on the initial state $\rho$, where $T_{A'}$ is the trivial POVM related to the probability distribution $\frac{1}{d} \tr{A'_i}$.
By choosing
\begin{align}\label{eq:choice}
A'_i = \frac{1}{c_d} \left[ A_i - \frac{1-c_d}{d} \tr{A_i} I \right]
\end{align}
the mixture $c_{d} A' + (1-c_{d}) T_{A'}$ is $A$.
The condition is hence that $A'$ in \eqref{eq:choice} is a valid POVM, meaning that
\begin{align*}
A_i \geq \frac{1-c_d}{d} \tr{A_i} 
\end{align*}
for each outcome $i$. This is equivalent to \eqref{eq:cloning-A}.
\end{proof}

\begin{remark}\label{rem:cloning-vs-noise-content-fixed-d}
The conditions \eqref{eq:cloning-A} and \eqref{eq:cloning-B} look similar to equation \eqref{eq:nc-criterion} from the noise content compatibility criterion. 
They are, however, incomparable at fixed dimension. 
To see this, let $A$ and $B$ be two qubit POVMs that correspond to two different bases $\{\varphi_1,\varphi_2\}$ and $\{\psi_1,\psi_2\}$.
We form two families of noisy versions of $A$ and $B$. 
Firstly, we define $A'$ and $B'$ as $A'_1=\half A_1$, $A'_2=\half A_2 + \half I$ and $B'_1=\half B_1$, $B'_2=\half B_2 + \half I$. The POVMs $A'$ and $B'$ satisfy the condition \eqref{eq:nc-criterion} but not \eqref{eq:cloning-A}--\eqref{eq:cloning-B}.
Secondly, we define $A''_1=\tfrac{2}{3} A_1 + \tfrac{1}{6} I$, $A''_2=\tfrac{2}{3} A_2 + \tfrac{1}{6} I$ and $B''_1=\tfrac{2}{3} B_1 + \tfrac{1}{6} I$, $B''_2=\tfrac{2}{3} B_2 + \tfrac{1}{6} I$. The POVMs $A''$ and $B''$ now satisfy the conditions \eqref{eq:cloning-A}--\eqref{eq:cloning-B} but not \eqref{eq:nc-criterion}.
\end{remark}

%%%%%%%%%%%%%%%%%%%%%%%%%%%%%%%%%%
\subsection{Miyadera-Imai criterion for incompatibility}\label{sec:MI-criterion}
%%%%%%%%%%%%%%%%%%%%%%%%%%%%%%%%%%

In \cite[Corollary 2]{miyadera2008heisenberg}, Miyadera and Imai provide a condition satisfied by all pairs of compatible POVMs. 
We recall it below (see also \cite[Section 3.2]{heinosaari2016invitation}).
We denote by $[\cdot, \cdot]$ the commutator, and $\sigma(\cdot)$ is the sharpness measure from Definition \ref{def:sharpness}.

\begin{proposition}\label{prop:MI-criterion}
If two POVMs $A$ and $B$ satisfy
$$
4 \|[A_i ,B_j]\|^2 > \sigma(A_i) \cdot \sigma(B_j),
$$
for all $i \in \Omega_A$ and $j \in \Omega_B$, then they are incompatible.
\end{proposition}

This condition covers as a special case the following well-known implication: if $A$ is sharp, then any POVM $B$ compatible with $A$ commutes with $A$.

%%%%%%%%%%%%%%%%%%%%%%%%%%%%%
\subsection{Zhu's criterion for incompatibility}
%%%%%%%%%%%%%%%%%%%%%%%%%%%%%

In the following we recall  Zhu's criterion \cite{zhu2015information} for detecting incompatible observables, which stems from the application of Gill-Massar inequality for Fisher information matrices \cite{gill2000state}. The criterion has a constructive approach, which we recall briefly, for the reader's convenience. 

Given  two POVMs $A$ and $B$, we define the superoperators $\mathcal{G}_A,\mathcal{ G}_B\in \mathcal {M}_d(\mathbb C)\otimes \mathcal {M}_d(\mathbb C)$ as
\begin{equation}\label{gs}
\mathcal{G}_A:=\sum\limits_{i} \frac{|A_i\rangle\langle A_i|}{\tr{A_i}}  \qquad \text{ and } \qquad
\mathcal{G}_B:=\sum\limits_{j}\frac{|B_j\rangle\langle B_j|}{\tr{B_j}},
\end{equation}
where $|A_i\rangle=\operatorname{vec}(A_i)\in \mathbb{C}^d\otimes\mathbb{C}^d$ denotes the vectorization (or flattening) of the matrix $A_i$: if $X = \sum_{i,j} x_{ij} e_ie_j^*$ is a matrix, then
$$|X\rangle = \operatorname{vec}(X) = \sum_{ij} x_{ij} e_i \otimes e_j,$$
for some orthonormal basis $\{e_i\}_{i=1}^d$ of $\mathbb C^d$. By denoting
\begin{equation}\label{t}
\tau(\mathcal{G}_A,\mathcal{G}_B)=\min_{H\geq \mathcal{G}_A,H\geq \mathcal{G}_B} \tr{H}, 
\end{equation}
Zhu established the following incompatibility criterion \cite[Equation (10)]{zhu2015information}. 

\begin{proposition}\label{prop:Zhu}
If $\tau(\mathcal{G}_A,\mathcal{G}_B)>d$, then the POVMs $A,B$ are incompatible.
\end{proposition}

It is clear that the quantity $\tau$ from \eqref{t} is the value of a semidefinite program \cite{boyd2004convex}. Indeed, we can associate to it the Lagrangian
$$\mathcal{L}(H,x,y)=\tr{H} + \langle x,\mathcal{G}_A-H \rangle + \langle y,\mathcal{G}_B-H\rangle $$
and define 
\begin{eqnarray}\label{cond}
g(x,y)&:=&\min\limits_H \mathcal{L}(H,x,y)=
\min_H \langle H, I-x-y\rangle + \langle x,\mathcal{G}_A \rangle +\langle y,\mathcal{G}_B \rangle \nonumber\\
&=& \begin{cases} \langle x,\mathcal{G}_A \rangle +\langle y,\mathcal{G}_B \rangle ,&\qquad \text{ if }  x+y=I\\
-\infty,&\qquad \text{ otherwise.}
\end{cases}
\end{eqnarray}
Furthermore, by associating the dual condition to \eqref{cond}, it follows that
\begin{equation}\label{cond1}
\max_{x,y \geq 0,\, x+y=I}\langle x,\mathcal{G}_A \rangle +\langle y,\mathcal{G}_B \rangle = \max_{0\leq x\leq I} \langle x,\mathcal{G}_A-\mathcal{G}_B \rangle + \tr{\mathcal{G}_B}.
\end{equation}
The optimal value for $x$ for \eqref{cond1} is achieved at $x_{opt}=P_+(\mathcal{G}_A-\mathcal{G}_B)$, the orthogonal projection on the eigenspaces corresponding to non-negative eigenvalues of $\mathcal{G}_A-\mathcal{G}_B$; one notices the similarity between this SDP and the one for optimal discrimination of quantum states \cite{helstrom1969quantum,holevo1973statistical}. We conclude:
\begin{equation}\label{eq:tau-Zhu-1-norm}
\tau(\mathcal{G}_A,\mathcal{G}_B)=\tr{(\mathcal{G}_A-\mathcal{G}_B)_{+}}+\tr{\mathcal{G}_B}=\frac{1}{2}\big[ \tr{\mathcal{G}_A}+\tr{\mathcal{G}_B}+\|\mathcal{G}_A-\mathcal{G}_B\|_1\big].
\end{equation}

%%%%%%%%%%%%%%%%%%%%%%%%%%%%%%%%%%%
\section{Interlude: random matrix theory and free probability}
\label{sec:interlude}
  This section aims to recall basic definitions and concepts necessary for a facile understanding of the current work, rendering it self-contained.
 The theory of Haar-distributed random unitary operators and isometries is reviewed, to be connected in the following sections to the theory of random quantum channels and random POVMs.  In addition, overviews on (graphical) Weingarten calculus and free probability are given. In each case, the main concepts are presented, and the theorems which shall be used later are stated without proofs; references are given for the reader interested in further exploring these topics.
 
\subsection{Random isometries and channels}
 
Let us recall here the notion of quantum channel in order to justify the study of random isometries. A \emph{quantum channel} is a linear map $\Psi:\mathcal{M}_d({\mathbb{C}})\rightarrow \mathcal{M}_k({\mathbb{C}})$ which is completely positive and trace preserving. Alternatively, using  the dual map with respect to the usual scalar product, the map $\Psi^*:\mathcal{M}_k({\mathbb{C}})\rightarrow \mathcal{M}_d({\mathbb{C}})$ is completely positive and unital.
Stinespring's representation theorem (see, e.g.~\cite[Chapter 4]{heinosaari2012mathematical} or \cite[Chapter 2.2]{watrous2018theory}) states that any quantum channel $\Psi$ can be written as 
\begin{equation}\label{eq:stinespring}
\Psi(X) = [\operatorname{id}_k \otimes \operatorname{Tr}_n](VXV^*), \qquad \forall X\in\mathcal{M}_d(\mathbb{C})
\end{equation}
where  $V:\mathbb{C}^d\rightarrow \mathbb{C}^k\otimes \mathbb{C}^n$ is an isometry. In the dual picture, we have the following representation of completely positive, unital maps
\begin{equation}\label{eq:stinespring-dual}
\Psi^*(Y)= V^*(Y \otimes I_n)V, \qquad \forall Y\in\mathcal{M}_k(\mathbb{C}).
\end{equation}
The Stinespring representation works also conversely: any isometry $V$ gives rise to a quantum channel. This fact 
is used to introduce random quantum channel, obtained by a random choice of the isometry $V$ in \eqref{eq:stinespring} or \eqref{eq:stinespring-dual}; we explain next what we call a \emph{random isometry}.

The set of all isometries $\{V:\mathbb{C}^d\rightarrow \mathbb{C}^D\}$ admits a \emph{unique} left- and right- invariant probability measure, called Haar measure, which can be obtained from the Haar measure on the unitary group $\mathcal{U}(kn)$ \cite[Section 4.2]{hiai2000semicircle} by truncation. More precisely, there is a unique probability measure $\mu_{Haar}$ on the set of isometries $\mathbb{C}^d \to \mathbb{C}^D$ which has the property that, if $V \sim \mu_{Haar}$, then, for all $U_1 \in \mathcal U(d)$ and $U_2 \in \mathcal U(D)$, the isometry $U_2 V U_1 \sim \mu_{Haar}$.

Using random isometries, \emph{random quantum channels} were introduced in \cite{hayden2008counterexamples} by choosing the isometry $V$ appearing in the Stinespring representation from the Haar ensemble. Indeed, for each pair of integers $d,k$, and for all values of the parameter $n$, the set of all channels $\{\Psi:\mathcal{M}_d({\mathbb{C}})\rightarrow \mathcal{M}_k({\mathbb{C}})\}$ is endowed with the measure induced by the Haar distribution on the set of isometries $V$ by the map \eqref{eq:stinespring} that associates to $V$ the channel $\Psi$. Although there are many other probability distributions on the set of quantum channels, in this paper we are going to be concerned with the one above. 
 
This model of random quantum channels has been used with great success in the theory of quantum information, starting with the work of Hayden and Winter \cite{hayden2008counterexamples}. Subsequently, several authors \cite{fukuda2010entanglement,aubrun2011hastingss,belinschi2016almost} have studied the application of this model of randomness to the problem of additivity of the minimum output entropy of quantum channels, see \cite[Section 6]{collins2016random} for a review.

 %%%%%%%%%%%%%%%%%
 \subsection{Weingarten formula}
 %%%%%%%%%%%%%%%%%
 
In order to compute properties of random quantum channels, one has to integrate over the set of Haar-distributed random isometries, or, equivalently, over the set of Haar-distributed unitary operators.
The expectation of products of entries of a random unitary operator has been considered in the physics literature by Weingarten in \cite{weingarten1978asymptotic} for the case of large matrix dimension. The rigorous mathematical analysis at fixed matrix size is due to Collins \cite{collins2003moments} and Collins-{\'S}niady \cite{collins2006integration}, where it was shown, using Schur-Weyl duality, that the moment integrals can be expressed as sums over the symmetric group.

\begin{theorem} Let $N$ be a positive integer  and $i=(i_1,\ldots, i_n), i'=(i'_1,\ldots, i'_n),j=(j_1,\dots, j_n), j'=(j'_1,\dots, j'_n)$ n-tuples of positive integers from $[N]=\{1,2,\ldots, N\}$. 
Let 
 $U\in \mathcal{U}(N)$ be an $N\times N$ Haar-distributed unitary random matrix  and denote by $U_{ij}$ the $(i,j)$-th entry of $U$ and $\delta_{ij}=\begin{cases} 1, i=j\\ 0,i\neq j\end{cases}$.
   Then, we have
\begin{equation}
\int\limits_{\mathcal{U}(N)}\label{int1}
U_{i_1j_1}\ldots U_{i_nj_n}\bar U_{{i_1'j_1'}}\ldots \bar U_{i_n'j_n'} \mathrm{d}U=
\sum\limits_{\alpha,\beta\in \mathcal{S}_n}\delta_{i_1i_{\alpha(1)}'}
\ldots  \delta_{i_ni_{\alpha(n)}'}\delta_{j_1j_{\beta(1)}'}
\ldots \delta_{j_nj_{\beta(n)}'} \Wg(N,\alpha^{-1}\beta),
\end{equation}
where the function $\Wg$ is called the \emph{Weingarten function}.
If $n\neq n'$, then
\begin{equation}\label{int2}
\int\limits_{\mathcal{U}(N)}
U_{i_1j_1}\ldots U_{i_nj_n}\bar U_{{i_1'j_1'}}\ldots \bar U_{i_{n'}'j_{n'}'} \mathrm{d}U=0
\end{equation}

\end{theorem}
The Weingarten function $\text{Wg}$ dates back to Weingarten \cite{weingarten1978asymptotic},  but the terminology and the notation were introduced by Collins \cite{collins2003moments}.
\begin{remark} For $\alpha\in \mathcal{S}_n$, $n\leq N$ and for $U\in \mathcal{U}(N)$ an $N\times N$ Haar-distributed unitary random matrix, where $\mathrm{d}U$ the normalized Haar measure, we have that
$$\Wg(N,\alpha)=\int\limits_{\mathcal{U}(N)}U_{11}\ldots U_{nn}\bar U_{1\alpha(1)}\ldots \bar U_{n\alpha(n)} \mathrm{d}U=\mathbb{E}[U_{11}\ldots U_{nn}\bar U_{1\alpha(1)}\ldots \bar U_{n\alpha(n)}].$$
\end{remark}
In the following we recall the definition of Weingarten function, give examples of it and present some of its properties used in the current paper.

\begin{definition}
The unitary Weingarten function $\Wg(N,\alpha)$, depending on the dimension parameter $N$ and on the permutation $\alpha$ in the symmetric group $\mathcal{S}_n$, is the inverse of the function $\alpha \mapsto N^{\# \alpha}$ under the following convolution operation for the symmetric group 
$$\forall \sigma,\pi\in \mathcal{S}_n, \qquad \sum\limits_{\tau\in\mathcal{S}_n}\Wg(N,\sigma^{-1}\tau)N^{\#(\tau^{-1}\pi)}=\delta_{\sigma,\pi}.$$
\end{definition}
The Weingarten function has the particularity that it depends only on the cycle structure of the permutation. For example, Wg$(N,[2,1])$ denotes the value of every permutation in $\mathcal{S}_3$ which  decomposition consists of a transposition and a fixed point. It holds that
$$\Wg(N,[2,1])=\frac{-1}{(N^2-1)(N^2-4)}.$$
More details related to the computation of Weingarten functions are given in \cite{collins2006integration}.
The dimension parameter in the notation of $\text{Wg}$ can be omitted when there is no confusion ($\text{Wg}(N,\alpha)\equiv \text{Wg}(\alpha)$). 
To the aim of our paper, it is of interest to present information about the behavior of Wg function in the large limit of $N$ (when $n$ is kept fixed). 
\begin{remark}
The asymptotics of Weingarten function is given by
\begin{equation*}
\Wg(N,\alpha)=N^{-(n+|\sigma|)}(\Mob(\alpha)+\mathcal{O}(N^{-2}),
\end{equation*} 
where the  M{\"o}bius function on the symmetric group is multiplicative with respect to the cycle structure of permutations:
\begin{equation*}
\Mob(\alpha)=\prod\limits_{c \text{ cycle of } \alpha}(-1)^{|c|-1}\mathrm{Cat}_{|c|-1}.
\end{equation*}
Here $\mathrm{Cat}_N$ is the N-th Catalan number.
In particular, if $\alpha$ is a product of disjoint transpositions, then
\begin{equation*}
\Mob(\alpha)=(-1)^{|\alpha|}
\end{equation*}
Frequently, we shall use the (justified) notation $\Mob(\alpha^{-1}\beta):=\Mob(\alpha,\beta)$.
\end{remark}

\subsection{Graphical calculus for random independent unitary matrices}
\label{Graphicalcalculus}
The integration formula \eqref{int1} used to evaluate expectation over Haar-distributed unitary random matrices usually involves sums indexed by large sets of indices, which  often turns out to be a complicated task to handle. In order to simplify tensors operations, the \emph{graphical Weingarten formalism} was introduced in \cite{collins2010random}. It builds up on Penrose's graphical tensor notation \cite{penrose1971applications} where diagrams consisting of boxes, decorations, and wires are used to represent tensors, collection of tensors, their dimensions, as well contraction operations on them. In \cite{collins2010random}, expectation values of diagrams $\mathcal D$ containing random, Haar-distributed unitary matrices $U$ and $\bar{U}$ are computed graphically, using the so-called \emph{removal} procedure. According to \eqref{int2}, if the number of $U$ boxes is different from the number of $\bar{U}$ boxes, then $\mathbb{E}\mathcal{D}=0$.
Otherwise,  we shall use a pair of permutation $(\alpha,\beta)\in \mathcal{S}_n^2$ to pair the decorations of the $n$ pairs of boxes $U/\bar{U}$. For each $i=1,\ldots, n$, wires are used to connect white decorations of the $k$-th $U$ box with the white decorations of the $\alpha(k)$-th $\bar{U}$ box. By a similar procedure  the black decorations are paired using now the $\beta$ permutation, see Figure \ref{fig:Wg-graphical}. The next step consists of erasing the $U/\bar{U}$ boxes and denoting by $\mathcal{D}_{\alpha,\beta}$ the resulting diagram. It holds that
\begin{equation}\label{rem}
\mathbb{E}_U(\mathcal{D})=\sum\limits_{\alpha,\beta} \mathcal{D}_{\alpha,\beta}\Wg(N,\alpha^{-1}\beta).
\end{equation}
\begin{figure}
	\centering
	\includegraphics[width=0.4\linewidth]{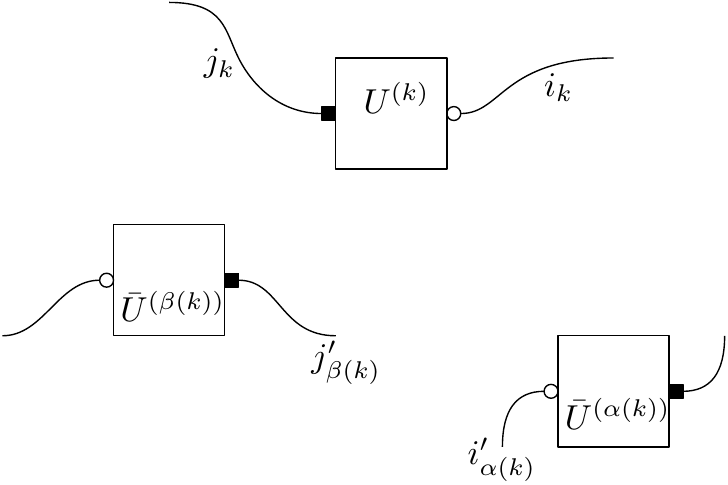}\qquad\qquad\qquad	\includegraphics[width=0.4\linewidth]{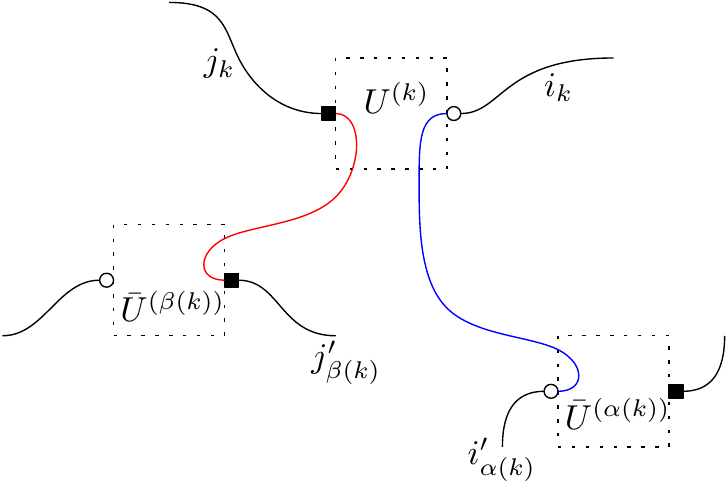}
	\caption{A part of a tensor diagram before and after the removal procedure. On the right panel, we add new wires to pair the decorations of the $U$ with the decorations of the $\bar U$ boxed according to the permutations $\alpha$ and $\beta$; we then delete the boxes corresponding to the Haar-distributed random unitary matrix $U$. }
	\label{fig:Wg-graphical}
\end{figure}
The formula above is just the interpretation of the algebraic expression \eqref{int1} in the tensor graphical calculus. Explicit examples for the use of \eqref{rem} are given in \cite{collins2010random}.

%%%%%%%%%%%%%%%%%%%%
\subsection{Tools from free probability}\label{sec:free-probability}
%%%%%%%%%%%%%%%%%%%%
This section aims to recall basic statements from free probability needed for a good understanding of the paper. We shall only sketch the concepts and results that shall be used later in the paper; we refer the reader to the monographs \cite{voiculescu1992free,nica2006lectures,mingo2017free} for the details. 

\par A \emph{$C^*$ probability space} is the pair $(\mathcal{A},\varphi)$, where  $\mathcal{A}$ is a unital $C^*$ algebra, with involution $a\mapsto a^*$, endowed with the state $\varphi$, i.e. $\varphi:\mathcal{A}\rightarrow \mathbb{C}, \varphi$-positive. The norm satisfies $\|a\|=\lim\limits_{p\rightarrow \infty}(\varphi(a^p))^{1/p}$. Given a selfadjoint element  $a$, the distribution of $a$, denoted by ${\mu}_a$, is the probability measure on the spectrum of  $a$, given by 
\begin{equation*}
\int x^pd{\mu}_a(x)=\varphi(a^p),\quad \forall p\in\mathbb{N}^*.
\end{equation*} 
The number $\varphi(a^p),p\in\mathbb{N}^*$ is called the $p$-th moment of $a$.  The moments of the random variable $a$ are usually identified to the moments of the probability measure $\mu_a$, which are given by $m_p(\mu_a):=\int x^pd{\mu}_a(x)$.
In this paper we are mostly concerned with the convergence of the eigenvalues of random matrices.  In  $C^*$ probability spaces, one can consider two types of convergence: the convergence in distribution (which is the convergence of all moments if, say, the limit measure has compact support) and the strong convergence (which implies, in particular, the convergence of the extreme eigenvalues of the matrices). It is of interest to recall the convergence in distribution does not imply strong convergence.  
\begin{definition}
Given the $C^*$ probability spaces  $(\mathcal{A},\varphi,\|\cdot\|_{\varphi})$ and $(\mathcal{A}^{(N)},\varphi_N,\|\cdot\|_{\varphi_N})$ with $N\in\mathbb{N}$, where $\varphi$ and $\varphi_N$ are faithful traces. 
For the $n$-tuple $a=(a_1,\ldots,a_n)\in\mathcal{A}$ and $a^{(N)}=(a_1^{(N)},\dots,a_n^{(N)})\in\mathcal{A}^{(N)}$, we say that 
\begin{itemize}
\item $a^{(N)}$ converges in distribution if
\begin{equation*}
\lim\limits_{N\rightarrow \infty}\varphi_N[P(a^{(N)},a^{(N)*})]=\varphi[P(a,a^*)]
\end{equation*}
\item $a^{(N)}$ converges to a strongly in distribution if, in addition, 
\begin{equation*}
\lim\limits_{N\rightarrow \infty}\|P(a^{(N)},a^{(N)*})\|_{\varphi_N}=\|P(a,a^*)\|_{\varphi}.
\end{equation*}
\end{itemize} 
\end{definition}
The theory of free probability is based on new concepts such as \textit{free independence, free cumulants, free convolution,} etc. In the following we recall some of them.
Given a probability measure $\mu$ on the real line with compact support, its free cumulants $k_p(\mu)$ are given by the moment-cumulant formula \cite{nica2006lectures}
\begin{equation}
m_p(\mu)=\sum\limits_{\pi\in NC(p)}\prod\limits_{b\in \pi}k_{|b|}(\mu).
\end{equation}
Obviously, the free-cumulants $k_{p}(\mu)$ contain the same information as the moments of the measure $m_p(\mu)$.

We recall that given two free elements $a,b$ having distributions $\mu,\nu$, the distributions of $a+b$ is denoted by $\mu\boxplus \nu$ and it is called the \emph{free additive convolution} of $\mu$ and $\nu$, see \cite[Lecture 12]{nica2006lectures}. Given the case of Bernoulli distributions $b_t=(1-t)\delta_0+t\delta_1$, it holds that (see \cite[Example 3.6.7]{voiculescu1992free} and \cite[Exercise 14.21]{nica2006lectures}).
\begin{proposition}\label{prop:free-additive-power-Bernoulli}
For any $T\geq 1$, the free additive power of a Bernoulli distribution is given by
\begin{align*}
b_s^{\boxplus T}&=\max(0,1-Ts)
\delta_0+\max(0,1-T(1-s))\delta_T+\\
&\qquad\qquad\frac{T\sqrt{({\gamma}^{+}(s,T)-x)(x-{\gamma}^{-}(s,T))}}{2\pi x(T-x)}\mathbf 1_{[\gamma^-(s,T),\gamma^+(s,T)]}(x)dx
\end{align*}
where $\gamma^{\pm}=(T-2)s+1\pm2\sqrt{(T-1)s(1-s)}$.
\end{proposition}

We recall below a lemma related to the push-forward property of the free additive convolution of probability measures. The following notation is used:  $f_{\#\mu}$ is the push-forward of a measure $\mu$ by a measurable function $f$; it holds that, given a random variable $X$ of distribution $\mu$, then $f(X)$ has distribution $f_{\#\mu}$.
\begin{lemma}Let $\mu$ be a compactly supported probability measure on $\mathbb{R}$ so that, for any $T\geq 1$, the distribution $\mu^{\boxplus T}$ is well-defined. Then, we have, for any $a,b\in\mathbb{R}$
$$((x\mapsto ax+b)_{\#\mu})^{\boxplus T}=(x\mapsto ax+Tb)_{\#}(\mu^{\boxplus T}).$$
\end{lemma}

%%%%%%%%%%%%%%%%%%%%%%%%%%%
\section{Random POVMs}\label{Secc:random-POVM}
%%%%%%%%%%%%%%%%%%%%%%%%%%%

This section contains one of the main contributions of this work, the definition and the basic properties of \emph{random POVMs}. We focus on one specific model, which we study in detail in the first subsection; this same model will be used in the rest of the paper to analyze the different quantities and (in-)compatibility criteria introduced in Sections \ref{sec:POVMs} and \ref{sec:compatibility}. In the second part of this section we consider an a priori different probability distribution over the set of POVMs of a given size with a given number of outcomes, obtained by normalizing independent Wishart random matrices. We then show, that in the range of parameters we are interested in, under some symmetry assumption, this Wishart-like distribution coincides with our main model. Finally, we briefly discuss other possibilities for defining random POVMs in the third and final subsection. 

\subsection{Haar-random POVMs}

Our approach to random POVMs comes from the observation that if $\Phi:\mathcal M_k(\mathbb C) \to \mathcal M_d(\mathbb C)$ is a unital and positive map, then the image of the diagonal projections $\{\ket{e_i}\bra{e_i}\}_{i=1}^k$ through $\Phi$ form a POVM:
$$M_i := \Phi(\ket{e_i} \bra{e_i}), \qquad 1\leq i \leq k.$$
Indeed, since the map $\Phi$ is positive, the POVM elements $M_i$ are positive semidefinite, and the total probability condition follows from the fact that $\Phi$ is unital: 
$$\sum_{i=1}^k M_i = \Phi\left(\sum_{i=1}^k \ket{e_i}\bra{e_i} \right)= \Phi(I_k) = I_d.$$
We are going to strengthen the requirements above and consider \emph{random unital, completely positive maps} $\Phi$ coming from Haar isometries. Such maps $\Phi$ are duals (for the Hilbert-Schmidt scalar product) of \emph{random quantum channels}; see the discussion at the beginning of Section \ref{sec:interlude}. As explained, choosing the isometry $V$ appearing in the formula \eqref{eq:stinespring} for the Stinespring dilation of a quantum channel (or its dual, see \eqref{eq:stinespring-dual}), one induces a random quantum channel. Let us note that a similar model of random POVMs has been introduced in \cite{radhakrishnan2009random} in relation to the hidden subgroup problem; there however the focus was on the distinguishability power of such measurements, and the analytical properties of the random POVMs were not investigated. 

\begin{definition}\label{def:Haar-random-POVM}
	Fix an orthonormal basis $\{e_i\}_{i=1}^k$ of $\mathbb C^k$ and consider a Haar-distributed random isometry $V:\mathbb C^d \to \mathbb C^k \otimes \mathbb C^n$, for some triple of integers $(d,k,n)$ with $d \leq kn$. Define the unital, completely positive map 
	\begin{align*}
	\Phi : \mathcal M_k(\mathbb C) &\to \mathcal M_d(\mathbb C)\\
	X &\mapsto V^*(X \otimes I_n)V.
	\end{align*}
	A \emph{Haar-random POVM} of parameters $(d,k;n)$ is the $k$-tuple $(M_1,\ldots, M_k)$ defined by
	$$\mathcal M_d(\mathbb C) \ni M_i := \Phi(\ket{e_i} \bra{e_i}).$$
\end{definition}
\begin{remark}
	In Definition \ref{def:Haar-random-POVM}, the parameter $n$ can be any integer satisfying $n \geq d/k$. However, the distribution of the random POVM $M$ makes sense for all \emph{real} values $n \in [d/k, \infty)$, see Remark \ref{rem:real-n}.
\end{remark}

\begin{remark}
	The POVM elements $M_i$ can be written as $$M_i	 = V_i^*V_i, \qquad 1\leq i \leq k,$$
	where $V_i$ are the $n \times d$ blocks of $V$: $V = \sum_{i=1}^k \ket{e_i} \otimes V_i$.
\end{remark}

\begin{remark}
A decomposition $M_i= V_i^*V_i$ defines an instrument of the POVM $M$. 
Namely, the map $\rho \mapsto V_i \rho V_i^*$ has the properties that 
\begin{equation*}
\tr{V_i \rho V_i^*} = \tr{\rho M_i}
\end{equation*}
and $\rho \mapsto \sum_i V_i \rho V_i^*$ is a quantum channel.
We remark that not all instruments of $M$ are of this form;
in general, an instrument $\mathcal{I}$ of $M$ is given as
\begin{equation*}
\mathcal{I}_i(\rho) = \sum_{j\in X_i} V_{ij} \rho V_{ij}^\ast \, ,
\end{equation*}
with 
\begin{equation*}
\sum_{j\in X_i} V_{ij}^\ast  V_{ij} = M_i \, ,
\end{equation*}
where $X_1,\ldots,X_k$ form a partition of some index set $\{1,\ldots,n\}$ into disjoint subsets. 
\end{remark}

\begin{remark}
	The same approach for constructing random POVMs is used in the function \emph{RandomPOVM} of the \texttt{QETLAB} library \cite{qetlab}, with the particular choice $n=d/k$. Our \texttt{MATLAB} routine\footnote{code for the \texttt{MATLAB} routine \texttt{RandomHaarPOVM.m} (as well as for other numerical functions used in this work) can be found in the supplementary material of the arXiv version of this paper} is more general, and allows for arbitrary integer values of $n$ (satisfying $n\geq d/k$).
\end{remark}

We gather in the next proposition some basic facts about Haar-random POVMs. 

\begin{proposition}
	Let $M = (M_1, \ldots, M_k)$ be a Haar-random POVM of parameters $(d,k;n)$. The random $k$-tuple $M$ is permutation invariant: for any permutation $\sigma \in \mathcal S_k$, the random variables 
	$$(M_1, \ldots, M_k) \quad \text{ and } \quad  (M_{\sigma(1)}, \ldots, M_{\sigma(k)})$$
	have the same distribution. In particular, the random matrices $\{M_i\}_{i=1}^k$ are identically distributed. Moreover, with probability one, the rank of a POVM element $M_i$ is $\min(d,n)$.
\end{proposition}
\begin{proof}
	The first assertion follows from the invariance of the Haar distribution of the random isometry $V$ from Definition \ref{def:Haar-random-POVM}: for any permutation $\sigma \in \mathcal S_k$, the isometries $V$ and $(P_\sigma \otimes I_n)V$ have the same distribution (here, $P_\sigma$ is the permutation matrix corresponding to $\sigma$).
	
	The second assertion follows from the fact that the rank of any sub-matrix of a Haar-distributed random unitary matrix is the minimum of its dimensions (i.e.~the maximum rank allowed). Indeed, if a sub-matrix had smaller rank, one could find a polynomial in the matrix entries which would vanish; it is a classical result in algebraic geometry (see, e.g.~\cite[Lemma 4.3]{nechita2012random} and the references within) that such a polynomial either vanishes on the whole unitary group or it vanishes on a set of measure zero. By constructing an explicit example, one can see that the former situation cannot happen, and the proof is complete.
\end{proof}

\begin{remark}
	In subsequent sections, we will vary the parameter $n$ to interpolate between POVMs having elements with small rank ($n \ll d$) and POVMs with invertible elements ($n \geq d$) allowing us to test the strength of various necessary (resp.~sufficient) conditions for compatibility found in the literature. One of the reasons we prefer this model of randomness for POVMs is the existence (at fixed $d,k$) of this 1-parameter family of probability measures. A similar framework was developed for the study of random quantum states, see \cite{zyczkowski2001induced} or \cite[Section IV.A.2]{collins2016random}.
\end{remark}

\subsection{Wishart-random POVMs}
We consider in this section another model of randomness than the one stemming from (duals of) random quantum channels. The starting point here is Wishart ensemble of Random Matrix Theory \cite{wishart1928generalised} (see also \cite[Chapter 3]{bai2010spectral} for a textbook introduction). Recall that a Wishart random matrix with parameters $(d,s)$ is given by $W=G^*G$, where $G \in \mathcal M_{s \times d}(\mathbb C)$ is a \emph{Ginibre random matrix}, that is a matrix with i.i.d.~complex standard Gaussian entries. Wishart matrices are, by construction, positive semidefinite, so one needs to apply a normalization procedure in order to construct a POVM. A similar model has been considered in \cite{aubrun2016zonoids}, where instead of normalizing independent Wishart matrices, the authors consider independent rank one projections (this imposes that the number of outcomes should be larger than the dimension).

 We summarize the construction in the following definition.
\begin{definition}
	A \emph{Wishart-random POVM} of parameters $(d,k;s_1, \ldots, s_k)$ is a $k$-tuple of matrices $(M_1, \ldots, M_k)$ where
	$$ M_i = S^{-1/2}W_i S^{-1/2}$$
	with $S = \sum_{i=1}^k W_i$ and $\{W_i\}_{i=1}^k$ is a family of independent Wishart matrices of respective parameters $(d,s_i)$ for $1 \leq i \leq k$.
\end{definition}
The Wishart-POVM ensemble might be useful in practice in the presence of an \emph{a priori} requiring different distributions for the POVM elements. Note that, in the case when $s_1 = \cdots = s_k$, the distribution of the POVM elements $\{M_i\}$ is permutation invariant. In fact, in the case where the common value of the parameters is an integer, the distribution of a Wishart-random POVM is exactly the distribution from Definition \ref{def:Haar-random-POVM}. 

\begin{theorem}\label{equiv}
The distribution of a Wishart-random POVM of parameters $(d,k;n, n, \ldots, n)$ is equal to the distribution of a Haar-random POVM of parameters $(d,k;n)$. 
\end{theorem}
\begin{proof}
Consider a Wishart-random POVM obtained from independent complex Gaussian matrices $G_1, \ldots, G_k \in \mathcal M_{n \times d}(\mathbb C)$. Stack the $G_i$ matrices on top of each other to form
$$G:= \sum_{i=1}^k |i \rangle \otimes G_i \in \mathcal M_{kn \times d}(\mathbb C).$$
The matrix $G$ is again a Gaussian matrix, since its entries are independent an follow a standard complex Gaussian distribution. Hence, its polar decomposition $G = VP$ can be chosen in such way that
\begin{enumerate}
	\item the positive part is $P = (G^*G)^{1/2} \geq 0$, with $P \in \mathcal M_d(\mathbb C)$
	\item the angular part $V: \mathbb C^d \to \mathbb C^{kn}$ is Haar distributed. 
\end{enumerate}
The latter condition follows from the unitary invariance of the Gaussian ensemble (note that $d \leq kn$). We have 
$$W_i = G_i^*G_i = G^*(I_n \otimes |i \rangle \langle i |) G,$$
and thus $S = \sum_{i=1}^k W_i = G^* G = P^2$. It follows that $S^{-1/2} = P^{-1}$ (where one might need to use the pseudo-inverse), and thus
$$M_i = S^{-1/2} W_i S^{-1/2} = P^{-1}G^* (I_n \otimes |i \rangle \langle i |) G P^{-1} = V^* (I_n \otimes |i \rangle \langle i |) V = V_i^* V_i,$$
where we have decomposed 
$$V= \sum_{i=1}^k |i \rangle \otimes V_i.$$
Since $V$ was chosen to be a Haar isometry, the conclusion follows.  
\end{proof}
\begin{remark}\label{rem:real-n}
Since Wishart matrices with parameters $(d,s)$ can be defined not only for integer $s$, but also for all real $s \geq d$, one can consider (Haar or Wishart)-random POVMs of parameters $(d,k;n)$ for any integers $d,k$ and 
$$n \in \left\{\left\lceil \frac d k \right \rceil, \left\lceil \frac d k \right \rceil+1 \ldots,d-1  \right\} \cup \left[ d , \infty \right).$$
\end{remark}

\begin{remark}
In practice, it is computationally cheaper to sample random Haar POVMs using Wishart matrices, than using Haar-distributed random isometries. However, from an analytical perspective, it is often more enlightening to use Definition \ref{def:Haar-random-POVM} of random POVMs. 
\end{remark}

\begin{remark}\label{rem:jacobi}
Let us also point out that the distribution of a single effect of a Wishart-random POVM is given by the Jacobi (or double Wishart) distribution from random matrix theory. Indeed. if $M$ is a Wishart-random POVM of parameters $(d,k;s_1,\ldots, s_k)$, then the random matrix $M_i$ has the same distribution as $(A+B)^{-1/2}A(A+B)^{-1/2}$, where $A$ has a Wishart distribution of parameters $(d,s_i)$ and $B$ is another Wishart matrix, independent from $A$, with parameters $(d,\check s_i)$, with $\check s_i = \sum_{j \neq i} s_j$.% Details about these distributions can be found, e.g., in \cite[Section 13.2]{anderson2003introduction}. Moreover, these distributions are related to the product of two random projections \cite{collins2005product}. 
\end{remark}

One can explicitly compute the density of a Wishart-random POVM with respect to the Lebesgue measure on $k$-tuples of Hermitian matrices using the matrix Dirac delta function \cite{hoskins2009delta, zhang2016dirac}; we defer the proof to the Appendix. 

\begin{theorem}\label{thm:density-Wishart-POVM}
The distribution of a Wishart-random POVM of parameters $(d,k;s_1, s_2, \ldots, s_k)$ has the following density at a point $m = (m_1, \ldots, m_k)$:
\begin{equation}\label{eq:density-Wishart-POVM}
\frac{\mathrm{d} \mathbb P}{\mathrm{d} \mathrm{Leb}}(m_1, \ldots, m_k) = C_{d,k,s_1, \ldots, s_k} \mathbf{1}_{\sum_j m_j = I_d} \prod_i \mathbf{1}_{m_i \geq 0} \det(m_i)^{s_i-d},
\end{equation}
where $C_{d,k,s_1, \ldots, s_k}$ is a normalization constant. 
\end{theorem}

\begin{corollary}\label{cor:Lebesgue}
In the particular case where $s_1 = \cdots = s_k = d$, the density above is flat, so one recovers the \emph{Lebesgue measure} on the set of POVMs. 
\end{corollary}

\subsection{Other distributions}

A third model of random POVMs comes from the notion of \emph{random bases}. Consider, for fixed $d$, a random basis $\{e_1, \ldots, e_d\}$ of $\mathbb C^d$, which can be obtained from the columns of a Haar-distributed, random unitary matrix $U$. 
For a mixing parameter $t \in [0,1]$, define the effect operators $M_i = t |e_i \rangle \langle e_i| + (1-t)I/d$, for all $i \in [d]$. 
This procedure defines a random $d$-outcome POVM in $\mathcal M_d(\mathbb C)$, depending on the parameter $t$. For $t=1$, we obtain a random von Neumann measurement on the vectors $e_i$, while for $t=0$ we get the trivial uniform POVM $(I/d, \ldots, I/d)$. Note that for this model of random POVMs, the number of outcomes is equal to the dimension of the effects. 

A fourth model is provided by the \emph{Lebesgue measure}. By Corollary \ref{cor:Lebesgue}, this measure is a special case of the parametric families we consider: we can obtain it either as a Haar-random POVM of parameters $(d,k;d)$  or as a Wishart-random POVM of parameters $(d,k;d,\ldots, d)$.

Finally, let us mention that random perturbations by Gaussian noise of a fixed given POVM have been considered in \cite{petz2012optimal} in a numerical algorithm used to find the optimal POVM for some particular state-estimation problem. 

%%%%%%%%%%%%%%%%%%%%%%%%%%%%
\section{Statistical properties of random POVMs}
\label{sec:stat-prop}
%%%%%%%%%%%%%%%%%%%%%%%%%%%%
We consider in this section the statistical properties of the effects $M_1, \ldots, M_k$ of a random POVM $M$, sampled from the ensemble introduced in the previous section, Definition \ref{def:Haar-random-POVM}. We shall be interested in the asymptotic spectrum of the individual effects $M_i$. These effects operators are elements of the \emph{Jacobi ensemble}, introduced by Wachter \cite{wachter1980limiting} and studied thoroughly in the random matrix theory literature \cite[Section 13.2]{anderson2003introduction},  \cite{johnstone2008multivariate}, \cite[Theorem 4.10]{bai2010spectral}, \cite[Section 3.6]{forrester2010log}. We use the graphical Weingarten calculus from \cite{collins2010random} to obtain moment formulas in a simple, combinatorial way, and then use free probability to re-derive the limiting spectral distribution.

\subsection{Exact moments of random effects}
 In the following proposition we aim to compute explicitly  the moments of a POVM element $M_i$ from the Haar-POVM ensemble; note that since the distribution of the random POVM $M$ is permutationally invariant, the value of $i$ is irrelevant, so we shall set $i=1$. We shall use the graphical Weingarten calculus introduced in Section \ref{Graphicalcalculus}.
\begin{proposition} 
\label{prop-mom}
For any integer dimensions parameters $n,d$, the moments of the random matrix $M_1\in\mathcal{M}_{nd}(\mathbb{C})$ are given by
\begin{equation}\label{mom}\forall p\geq 1, \qquad
\mathbb{E}\operatorname{Tr} M_1^p=\sum\limits_{\alpha,\beta\in\mathcal{S}_p} n^{\#\alpha}d^{\# (\beta\gamma^{-1})} \operatorname{Wg}(kn,\alpha^{-1}\beta)
\end{equation}
\end{proposition}
\begin{proof}
Without loss of generality, we can replace the random isometry $V$ in the definition of a random POVM by a random Haar-distributed unitary matrix $U\in\mathcal{U}_{kn}$. We aim to compute, for $\forall p\geq 1$, the moment  $\mathbb{E} \operatorname{Tr} M_1^p$; using indices, this reads
$$\mathbb{E} \operatorname{Tr} M_1^p = \mathbb E \sum_{i_1, \ldots, i_p=1}^d M_1(i_1,i_2) M_1(i_2, i_3) \cdots M_1(i_p, i_1).$$

\begin{figure}
	\centering
	\includegraphics[width=0.4\linewidth]{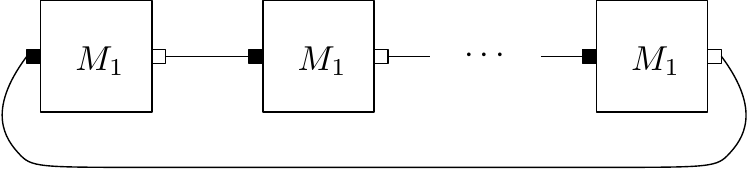}
	\caption{The diagram corresponding to the moment $\mathbb E \operatorname{Tr} M_1^p$. There are $p$ copies of the box $M_1$. The square labels attached to the boxes correspond to the space $\mathbb C^d$.}
	\label{fig:diagram-moment}
\end{figure}

In graphical notation, we aim to compute the expectation of the diagram $\mathcal D$ in Figure \ref{fig:diagram-moment}. We use the formula \eqref{rem} to compute the expectation value with respect to the random unitary matrix $U$. We use the removal algorithm, which assumes the rules recalled below:
\begin{itemize}
\item replace $U^*$ boxes by $\bar{U}$, as the removal procedure is requiring to pair decorations of the same color; the resulting diagram is presented in Figure \ref{fig:M1}
\item round decorations correspond to $\mathbb{C}^n$, whereas the square ones correspond to $\mathbb{C}^d$. Diamond shaped decorations correspond to $\mathbb C^k$, but they are not important in what follows since their contribution will be trivial
\item we aim to wire $p$ groups of $(U,\bar{U})$
\item using formula \eqref{rem}, the expectation of the diagram is a weighted sum (with Weingarten weights) of diagrams $\mathcal{D}_{\alpha,\beta}$, obtained after the removal of $U$ and $\bar{U}$.
\item the loops in the diagram are of two types: the ones connecting round decorations(each having a value of $n$) and the others are connecting square decorations (each having a value of $d$).
\end{itemize}

\begin{figure}
\centering
\includegraphics[width=0.4\linewidth]{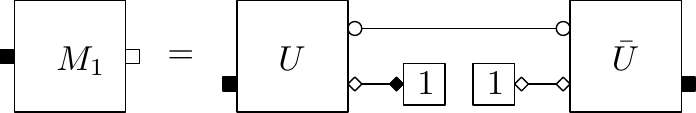}
\caption{The diagram for the random matrix $M_1$. We just write $1$ for the basis element $e_1 \in \mathbb C^k$.}
\label{fig:M1}
\end{figure}

 In consequence, the diagram $\mathcal{D}_{\alpha,\beta}$ consists of a collection of loops that correspond to different vector spaces, as follows:
 \begin{itemize}
 \item $\# \alpha$ loops of dimension $n$, corresponding to the round-shaped white labels. These decorations are actually connected to the identity permutation (in the original diagram) and the graphical expansion connects them by $\alpha$. The resulting number of loops is $\#\alpha =  \#(\alpha \cdot \mathrm{id}^{-1})$
 \item $\#(\beta\gamma^{-1})$ loops of dimension $d$, corresponding to square-shaped black labels. The square decorations are initially connected with the permutation $$\gamma:= (p\, p-1\, \cdots \, 3\, 2\, 1) \in \mathcal S_p$$
 that allows links of the form $l\rightarrow l-1$ and the graphical expansion connects them with the permutation $\beta$. The total number of loops is $\#({\beta}\gamma^{-1})$.
 \end{itemize}

Putting together the contributions above, weighted by the Weingarten factors, we obtain the claimed formula. 
\end{proof} 

Let us consider now the simplest cases of the formula in the result above, $p=1$ and $p=2$ respectively. 

At $p=1$, there is only one term in the sum, and we obtain 
\begin{equation}\label{eq:moment-1-exact}
\mathbb{E}\operatorname{Tr} M_1 = dn \operatorname{Wg}(kn, (1)) = \frac{dn}{kn} = \frac d k.
\end{equation}
This result was to be expected, since we know that 
$$d = \operatorname{Tr} I_d = \sum_{i=1}^k \mathbb E \operatorname{Tr} M_i = k\mathbb E \operatorname{Tr} M_1.$$ 
 
For $p=2$, the result is already non-trivial. We have that $\mathbb{E} \operatorname{Tr} M_1^2$ is a sum of four terms, corresponding to $\alpha,\beta \in \{\mathrm{id}, (1\, 2)\}$; the corresponding diagrams $\mathcal D_{\alpha,\beta}$ are depicted in Figure \ref{fig:ETrM2}. The terms are as follows: the wiring $\alpha=\beta=\mathrm{id}$ gives a contribution of $n^2d\operatorname {Wg}(kn,\mathrm{id})$. When $\alpha=\mathrm{id}$ and $\beta$ is the transposition $(1\, 2)$, we get the term $n^2 d^2\operatorname{Wg}(kn,(1 \, 2))$; but, if $\alpha=(1 \, 2)$ and $\beta=\mathrm{id}$, the contribution to the sum is $nd \operatorname{Wg}(kn, (1\, 2))$. The final situation, corresponding to $\alpha=\beta=(1\, 2)$ yields the term $nd^2\operatorname{Wg}(nk, \mathrm{id})$.
In conclusion, the total sum reads 
$$\mathbb{E} \operatorname{Tr} M_1^2=(n^2d+nd^2)\operatorname{Wg}(kn, \mathrm{id})+(n^2d^2+nd)\operatorname{Wg}(kn, (1\, 2)).$$
Using the corresponding values for the Weingarten functions
$$\operatorname{Wg}(kn, \mathrm{id}) = \frac{1}{(kn)^2-1} \qquad\qquad \operatorname{Wg}(kn, (1\, 2)) =  \frac{-1}{kn((kn)^2-1)},$$
it follows that 
\begin{equation}\label{eq:moment-2-exact}
\mathbb{E} \operatorname{Tr} M_1^2=(n^2d+nd^2)\frac{1}{(kn)^2-1}+(n^2d^2+nd)\frac{-1}{kn((kn)^2-1)} = \frac{d(kn^2+dn(k-1)-1)}{k((kn)^2-1)}.
\end{equation}
A similar computation gives the covariance between two different random POVM elements $M_1$ and $M_2$. The diagram for $\tr{M_1M_2}$ consists of the product of two copies of the diagram in Figure \ref{fig:M1}, with the ``1'' replaced by a ``2'' in the second copy. This fact imposes the constraint $\beta = \mathrm{id}$ in the Weingarten sum: terms with $\beta = (12)$ are zero because of the scalar product $\langle e_1, e_2\rangle$. We have thus:
\begin{equation}\label{eq:covariance-exact}
\mathbb{E} \tr{ M_1 M_2}= n^2d \Wg(kn, \mathrm{id}) + n^2d^2 \Wg(kn, (1,2)) = \frac{nd(kn-d)}{k((kn)^2-1)}.
\end{equation}

\begin{figure}
\centering
\includegraphics[width=0.45\linewidth]{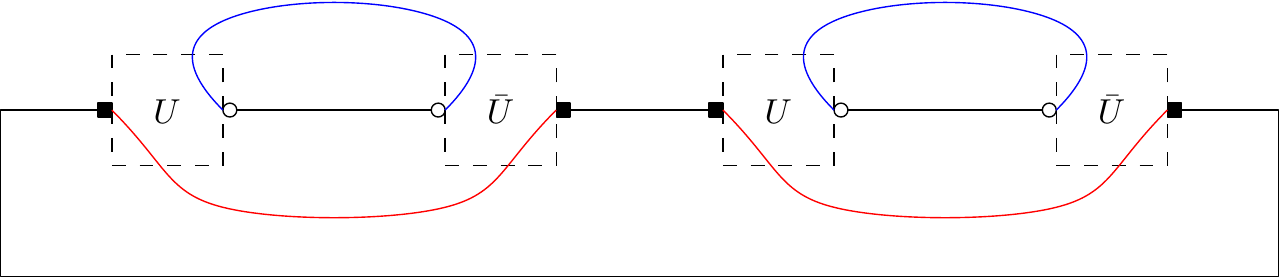}\qquad \includegraphics[width=0.45\linewidth]{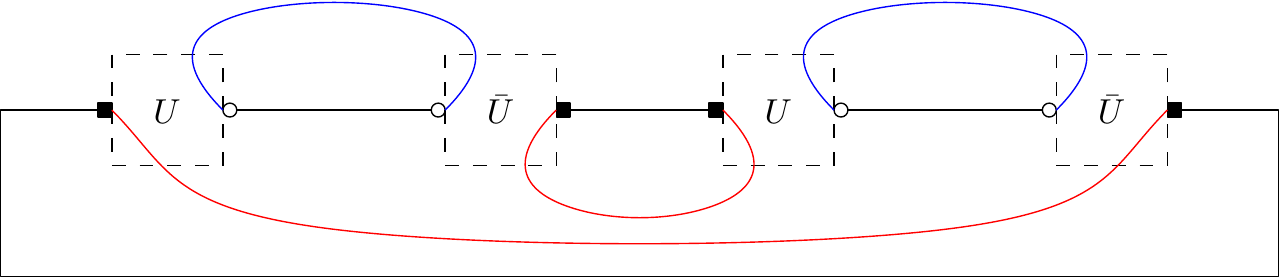}\\
\vspace{.5cm}
\includegraphics[width=0.45\linewidth]{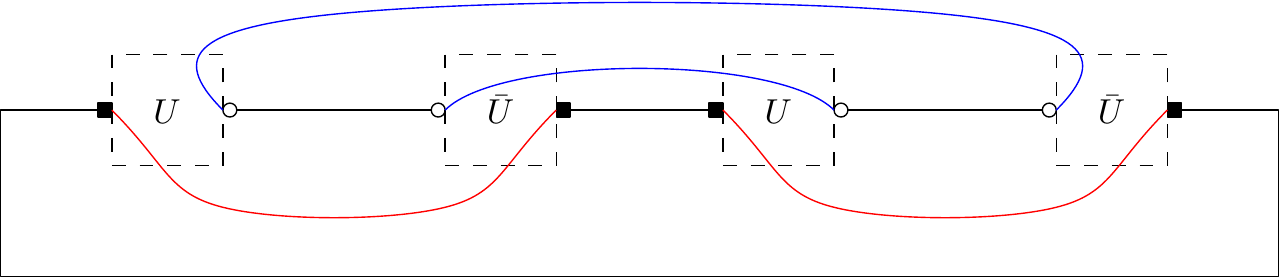}\qquad \includegraphics[width=0.45\linewidth]{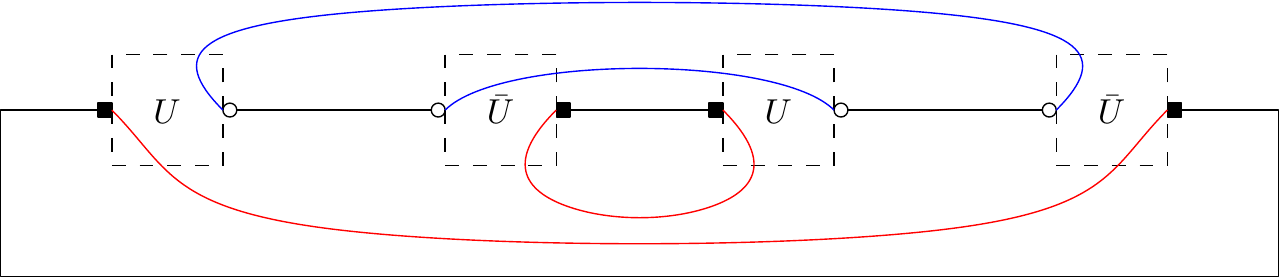}
\caption{The four diagrams appearing in the graphical expansion of $\mathbb{E} \operatorname{Tr} M_1^2$. From top to bottom, left to right, the diagrams correspond to $(\alpha,\beta)$ $=$ $(\mathrm{id},\mathrm{id})$, $(\mathrm{id},(1\, 2))$, $((1\, 2),\mathrm{id})$, $((1\, 2),(1\, 2))$. The permutation $\alpha$ is drawn on top, in blue, while $\beta$ is drawn downwards, in red. The value of each diagram is given by the number of loops in blue/red.}
\label{fig:ETrM2}
\end{figure}

\subsection{The asymptotical spectral distribution of random POVM effects}
With the help of free probability theory, we give here a simple derivation of the formula for the distribution of a POVM element $M_i$ from the Haar-POVM ensemble, in the large $d$ limit; for different approaches see e.g.~\cite[Theorem 4.10]{bai2010spectral}. To be more precise, we shall consider the following asymptotical regime: 
\begin{itemize}
\item $k$, the number of outcomes of the POVM is fixed
\item $d$, the dimension of the POVM effects grows to infinity
\item $n$, the parameter appearing in the definition of Haar-random POVMs grows to infinity, is such a way that
$$\lim_{d \to \infty} \frac{d}{kn} = t,$$
where $t \in (0,1]$ is a constant.
\end{itemize}

\begin{proposition}\label{prop:limit-eigenvalues-Mi}
	In the asymptotical regime where $k$ is fixed and $d,n \to \infty$ in such a way that $d \sim tkn$ for some constant $t \in (0,1]$, 
	the distribution of a POVM element $M_i$ from the Haar-POVM ensemble of parameters $(d,k;n)$ converges in moments towards the probability measure 
	\begin{equation}\label{eq:limit-measure-Mi}
	\begin{aligned}
	D_t\left [b_{k^{-1}}^{\boxplus t^{-1}}\right ] &= \max(0, 1-t^{-1}k^{-1})\delta_0 + \max(0, 1-t^{-1}+t^{-1}k^{-1})\delta_1 \\
	&\qquad + \frac{\sqrt{(x-\varphi_-)(\varphi_+-x)}}{2 \pi t x (1-x)} \mathbf 1_{[\varphi_-,\varphi_+]}(x) \mathrm{d}x,
	\end{aligned}
	\end{equation}
	where
	\begin{equation}\label{eq:phi-pm}
	\varphi_\pm = t + k^{-1} - 2tk^{-1} \pm 2 \sqrt{t(1-t)k^{-1}(1-k^{-1})}
	\end{equation}
	Above, $D_\cdot$ is the dilation operator (if $X$ has distribution $\mu$, then $aX$ has distribution $D_a \mu$), $b$ is the Bernoulli distribution ($b_p = (1-p) \delta_0 + p \delta_1$), and $\boxplus$ is the free additive convolution operation defined in Section \ref{sec:free-probability}. 
	
Moreover, the convergence also holds \emph{strongly}, in the sense of \cite{collins2014strong}. In particular, the extremal eigenvalues of $M_i$ converge almost surely to the edges of the support of the measure from \eqref{eq:limit-measure-Mi}. 
\end{proposition}

\begin{proof}
	The result follows from the large $d,n$ limit of the formula \eqref{mom}.  We shall study the terms which contribute asymptotically and then we shall identify the limiting probability distribution with the help of its free cumulants.

To this end, we recall that the Weingarten function  $\operatorname{Wg}(nk, \alpha^{-1}\beta)$ may be approximated to second order by ${(nk)}^{-p-|\alpha^{-1}\beta|}\Mob(\alpha^{-1}\beta)$, for permutations $\alpha,\beta \in \mathcal S_p$.
Consequently, the moments behave as
$$d^{-1}\mathbb{E}\operatorname{Tr} M_i^p\sim (ntk)^{-1}\sum\limits_{\alpha,\beta\in\mathcal{S}_p} n^{\#\alpha}(ntk)^{\# (\beta\gamma^{-1})} {(nk)}^{-p-|\alpha^{-1}\beta|}\Mob(\alpha,\beta).$$
Above, the only non-vanishing terms, as $d,n\rightarrow \infty$, are the ones containing the largest power of $n$. 
A straightforward analysis shows that
\begin{equation*}\text{power of }n=
-1+\# \alpha+\#(\beta\gamma^{-1})-p-|\alpha^{-1}\beta|=p-1-(|\alpha|+|\alpha^{-1}\beta|+|\beta^{-1}\gamma|)\leq 0,
\end{equation*}
where we have used the relation $|\alpha|=p-\# \alpha$ and the triangle inequality
$$|\alpha|+|\alpha^{-1}\beta|+|\beta^{-1}\gamma| \geq |\gamma|=p-1.$$
The above inequality is saturated if and only if the both $\alpha$ and $\beta$ lay on the \emph{geodesic} between 
the identity permutation and the full cycle $\gamma$; we write $\mathrm{id}- \alpha-\beta- \gamma$. Here, the notion of geodesic is in relation to the following distance function on the symmetric group $\mathcal S_p$:
$$\operatorname{dist}(\sigma, \pi) := |\sigma^{-1}\pi|.$$
We say that a permutaion $\chi$ lies on the geodesic between $\sigma$ and $\pi$ if $\chi$ saturates the triangle inequality
$$\operatorname{dist}(\sigma, \chi) + \operatorname{dist}(\chi, \pi) \geq \operatorname{dist}(\sigma, \pi).$$
Hence, we obtain the asymptotic moments 
\begin{equation*}
\lim_{n\rightarrow \infty} d^{-1}\mathbb{E} \operatorname{Tr}  M_i^p= 
\sum\limits_{id-\alpha-\beta-\gamma}t^{-1+\#(\beta^{-1}\gamma)}k^{-1-|\beta^{-1}\gamma|-|\alpha^{-1}\beta|}\Mob(\alpha,\beta).
\end{equation*}

Using the fact that, for geodesic permutations $\alpha,\beta$, $-1-|\alpha^{-1}\beta|-|\beta^{-1}\gamma|=-p+|\alpha|=-\#\alpha$, 
the equation above may be rewritten as 
\begin{equation}\label{mom2}
\sum\limits_{id- \alpha-\beta- \gamma} t^{-1+\#(\beta^{-1}\gamma)}k^{-\#(\alpha)}\Mob(\alpha,\beta) =\sum\limits_{id- \alpha-\beta- \gamma}t^{p-\#(\beta)}k^{-\#(\alpha)}\Mob(\alpha,\beta).
\end{equation}
  
We now fix $\beta\in\mathcal{S}_p$ and use the moment-cumulant formula \cite{nica2006lectures} in free probability to write 
$$\sum_{\id-\alpha-\beta}
k^{-\#(\alpha)}\Mob(\alpha,\beta)=\sum\limits_{id- \alpha-\beta} m_{p}(b_{k^{-1}})\Mob(\alpha,\beta)=\mathcal{K}_{\beta}(b_{k^{-1}}),$$
where $b_{k^{-1}}$ is the Bernoulli distribution $b_{k^{-1}}=(1-k^{-1})\delta_0+k^{-1}\delta_1$.

Therefore, equation \eqref{mom2} becomes
$$\sum\limits_{id- \beta- \gamma} t^{p-\#(\beta)}\mathcal{K}_{\beta}(b_{k^{-1}})=
t^p\sum\limits_{id- \beta- \gamma} t^{-\#(\beta)}\mathcal{K}_{\beta}(b_{k^{-1}})=
t^{p}m_p(b_{k^{-1}}^{\boxplus t^{-1}})=m_p\left(D_t[b_{k^{-1}}^{\boxplus t^{-1}}]\right),$$
proving the first claim. In the following we aim to express the distribution $D_t[b_{k^{-1}}^{\boxplus t^{-1}}]$	in the form presented in the statement of the theorem. 
Indeed, using Proposition \ref{prop:free-additive-power-Bernoulli}, we get that
\begin{align*}
D_t[b_{k^{-1}}^{\boxplus t^{-1}}]&=\{x \mapsto tx\}_{\#}(b_{k^{-1}}^{\boxplus t^{-1}}) =
\max(0,1-t^{-1}k^{-1})\delta_0+\max(0,1-t^{-1}(1-k^{-1}))\delta_1+\\
&\qquad \qquad \frac{1/t\sqrt{(\gamma^{+}-\frac{x}{t})(\frac{x}{t}-\gamma^{-})}}{2\pi \frac{x}{t}(\frac{1}{t}-\frac{x}{t})}\mathbf{1}_{[\gamma^{-},\gamma^{+}]}(\frac{x}{t})\frac{\mathrm{d}x}{t}
\end{align*}
where $\gamma^{\pm}(1/k,1/t)=(\frac{1}{t}-2)\frac{1}{k}+1\pm2\sqrt{(\frac{1}{t}-1)\frac{1}{k}(1-\frac{1}{k})}$.
By denoting $t\gamma^{\pm}(1/k,1/t)=\varphi^{\pm}(1/k,t)$, we obtain the result announced in \eqref{eq:limit-measure-Mi}.

The strong convergence follows from the strong asymptotic freeness results of Collins and Male \cite[Theorem 1.4]{collins2014strong} applied to the Haar-distributed random unitary matrices $U_n$ and a sequence of deterministic projections.
\end{proof}

\begin{remark}
	For $t=1$, the measure in the theorem is the Bernoulli measure $b_{k^{-1}}$.
\end{remark}
\begin{remark}
Since the probability distribution \eqref{eq:limit-measure-Mi} can have Dirac masses at 0 or 1 (never at both end points), its support may be non-convex. This happens whenever one or the other Dirac mass is present, that is when $t<1/k$ (Dirac mass et 0) or when $t>1-1/k$ (Dirac mass at 1).
\end{remark}
\begin{remark}
In light of the results from \cite{collins2005product}, the distribution above is equal to the free multiplicative convolution of two Bernoulli distributions, of parameters $1/k$ and $t$ respectively. We do not discuss this equivalent point of view here. 
\end{remark}

We now present some immediate consequences of the theorem above. These results are about quantities of interest in quantum information theory, such as regularity or the norm-1 property; we refer the reader to Section \ref{sec:spectral-properties-POVMs} for the definitions.

\begin{proposition}
	In the asymptotical regime where $k$ is fixed and $d,n \to \infty$ in such a way that $d \sim tkn$ for some constant $t \in (0,1]$, the first two limiting moments of the random effects $M_i$ read
\begin{align*}
\lim_{n \to \infty} \frac 1 d \mathbb E \tr{M_i} &= \frac 1 k\\
\lim_{n \to \infty} \frac 1 d \mathbb E \tr{M_i^2} &= \frac{tk+1-t}{k^2},
\end{align*}
while the asymptotic covariance of two different effects ($i \neq j$) behaves like
$$\lim_{n \to \infty} \frac 1 d \mathbb E \tr{M_i M_j} = \frac{1-t}{k^2}.$$
\end{proposition}
\begin{proof}
The first two formulas follow either from Proposition \ref{prop:limit-eigenvalues-Mi} for $p=1,2$ or from taking the limit in \eqref{eq:moment-1-exact} and \eqref{eq:moment-2-exact}. The covariance formula follows from equation \eqref{eq:covariance-exact}.
\end{proof}

\begin{proposition}
	In the asymptotical regime where $k$ is fixed and $d,n \to \infty$ in such a way that $d \sim tkn$ for some constant $t \in (0,1]$, a random POVM $M$ is \emph{regular} (see Definition \ref{def:regular}) iff 
	$$t \in \left( \frac 1 2 - \frac{2\sqrt{k-1}}{k}, \frac 1 2 + \frac{2\sqrt{k-1}}{k} \right) \, .$$
\end{proposition}
\begin{proof}
	The condition from the statement is equivalent to asking that $1/2$ is not an element of the support of the limiting spectral distribution of the random effects \eqref{eq:limit-measure-Mi}.
\end{proof}

\begin{proposition}
	In the asymptotical regime where $k$ is fixed and $d,n \to \infty$ in such a way that $d \sim tkn$ for some constant $t \in (0,1]$, a random POVM $M$ has the norm-1 property (see Definition \ref{def:norm1}) iff $t>1-1/k$. 
\end{proposition}
\begin{proof}
	This follows from \eqref{eq:limit-measure-Mi}, by asking that the weight of the Dirac mass $\delta_1$ is positive. 
\end{proof}

We display Monte Carlo simulations of a Haar-random POVM element, together with the theoretical curve from the theorem above in Figure \ref{fig:ev-M}. Different statistical properties of these POVM elements will be analyzed in subsequent sections. 

\begin{figure}[htbp]
	\begin{center}
		\includegraphics[scale=0.50]{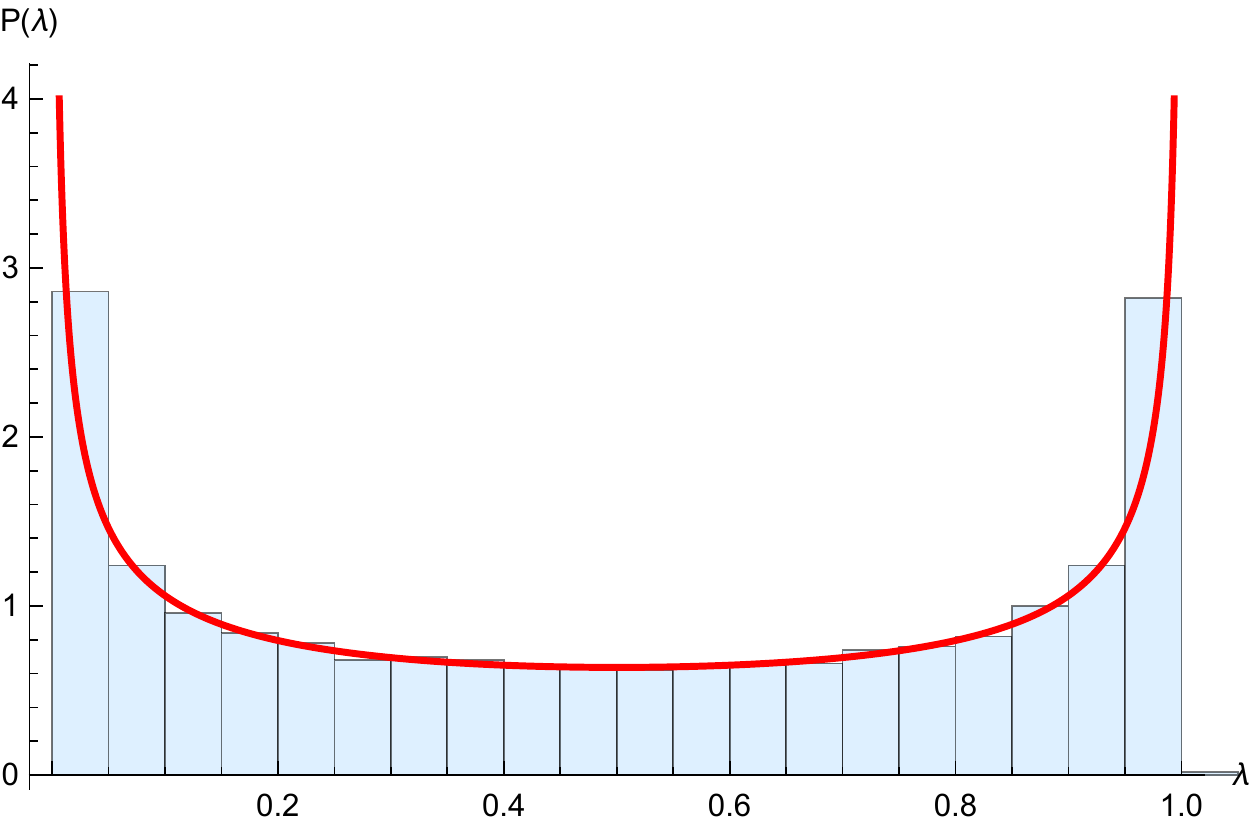} \qquad  \includegraphics[scale=0.50]{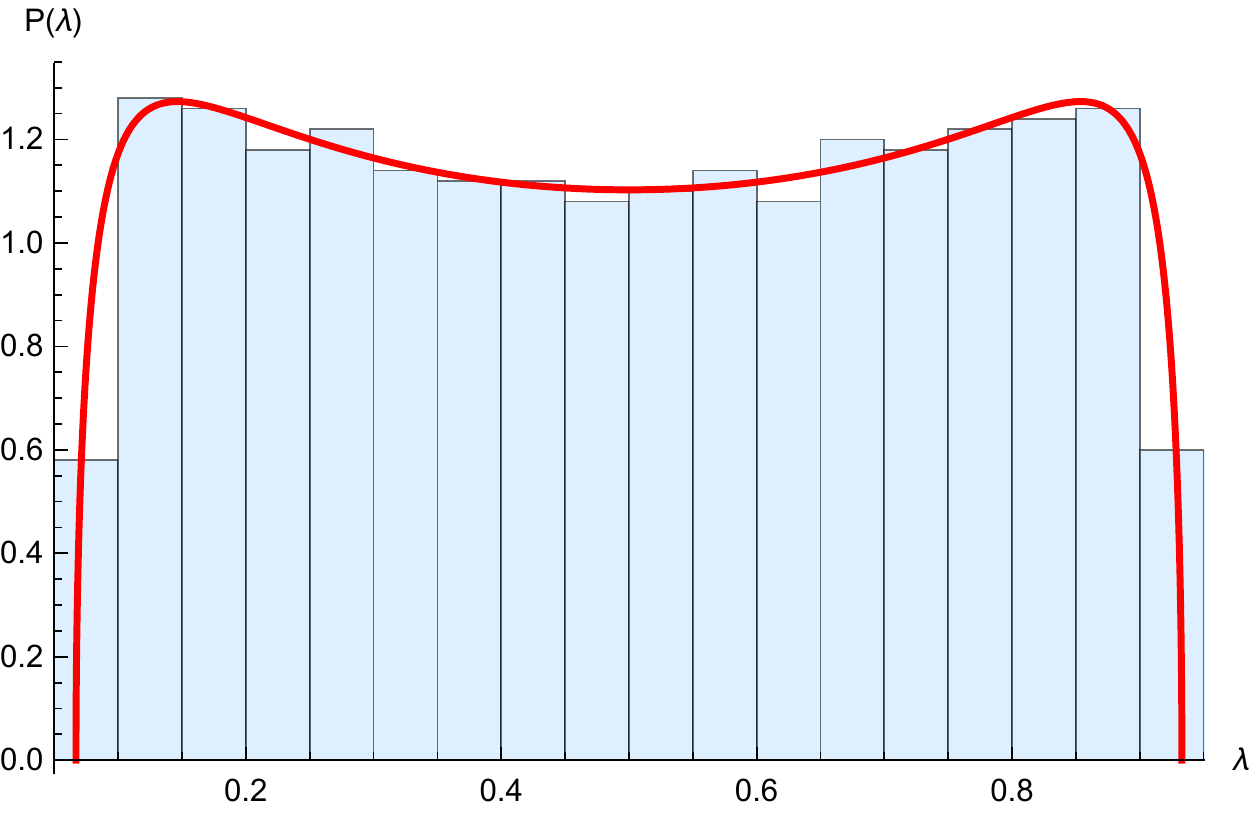} \\
		\includegraphics[scale=0.50]{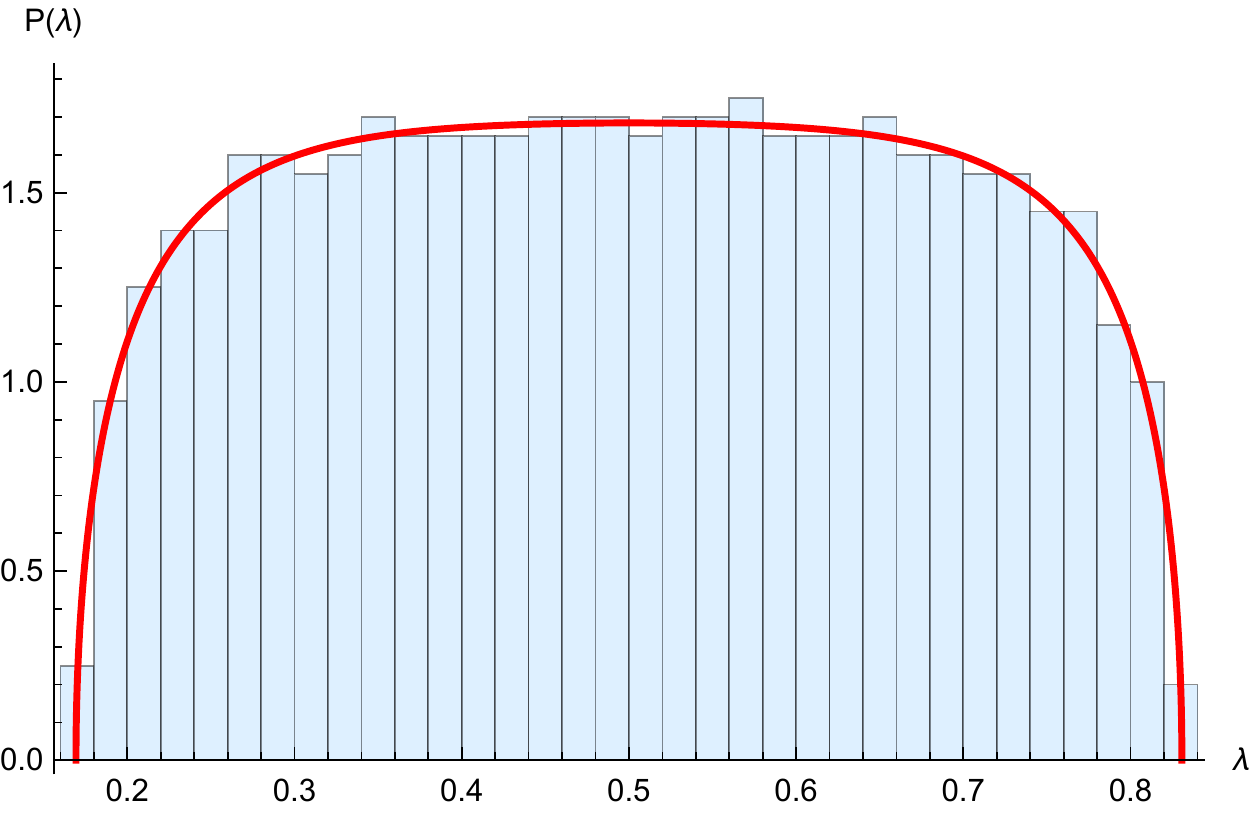} \qquad  \includegraphics[scale=0.50]{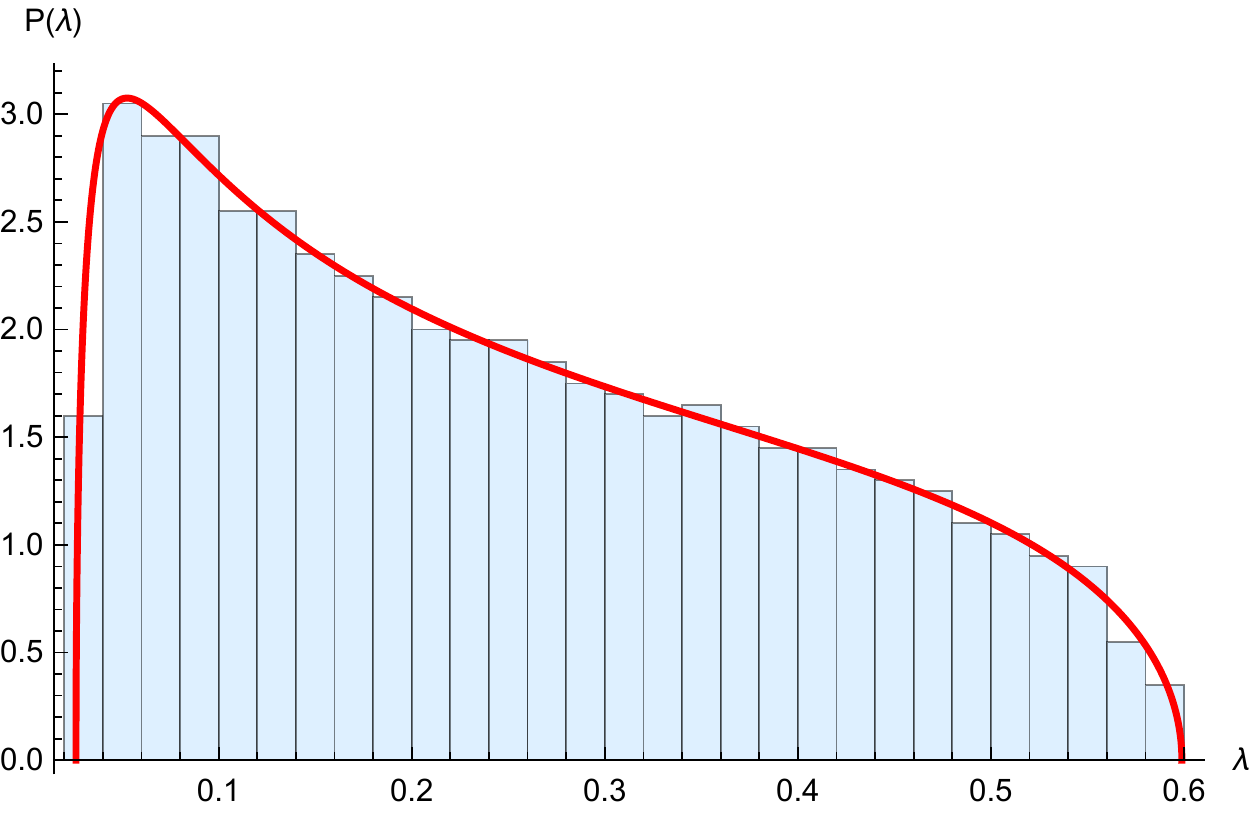} 		
		\caption{Monte Carlo simulations vs.~theoretical curves for the eigenvalues of Haar-POVM elements with the following $(d,k;n)$ triples: top-left (1000,2;1000), top-right (1000,2;2000), bottom-left (1000,2;4000), bottom-right (1000,4;2000). Since the first three examples are dichotomic POVMs, the plots are symmetric with respect to $x=1/2$. The histogram from each plot corresponds to the eigenvalues of a single sample.}
		\label{fig:ev-M}
	\end{center}
\end{figure}

%%%%%%%%%%%%%%%%%%%%%%%%%%%
\subsection{The probability range of random POVMs}
%%%%%%%%%%%%%%%%%%%%%%%%%%%

We now discuss the probability range of random POVMs. 
Since there is a close connection between the probability range and the output set of unital, completely positive maps, we shall use the results from \cite{belinschi2012eigenvectors} in the latter setting to obtain a characterization of the asymptotic probability range in the large dimension limit. Before we do this, let us provide a heuristic argumentation for Theorem \ref{thm:prob-ran-random-POVM}. Consider a random quantum channel 
\begin{align*} 
\Psi : \mathcal M_d(\mathbb C) \to \mathcal M_k(\mathbb C) \, , \quad \Psi(\rho) = [\operatorname{id}_k \otimes \operatorname{Tr}_n](V\rho V^*) \, .
\end{align*}
where $V : \mathbb C^d \to \mathbb C^k \otimes \mathbb C^n$ is a Haar-distributed random isometry. We know from Section \ref{Secc:random-POVM} that a random POVM has effects $M_i = \Psi^*(|i \rangle \langle i |)$, where $\Psi^*$ is the Hilbert-Schmidt adjoint of $\Psi$. Using this duality, we have
\begin{equation}\label{eq:duality}
[\operatorname{Tr}(\rho M_i)]_{i=1}^k = [\operatorname{Tr}(\rho \Psi^*(|i \rangle \langle i |))]_{i=1}^k =  [\operatorname{Tr}(\Psi(\rho) |i \rangle \langle i |)]_{i=1}^k = \operatorname{diag}\Psi(\rho).
\end{equation}

First, note that, given an arbitrary fixed pure quantum state $\ket{\psi} \in \mathbb C^d$, the distribution of the (random) probability vector 
$$(\langle \psi |M_1 | \psi\rangle, \ldots, \langle \psi |M_k | \psi\rangle) \in \Delta_k$$
is the \emph{Dirichlet distribution} of parameter $n$
$$\operatorname{Dir}_k^{(n)}(p_1, \ldots, p_k) \sim p_1^{n-1}p_2^{n-1}\cdots p_k^{n-1}.$$
Indeed, this follow from \eqref{eq:duality} and the fact that the diagonal of a random density matrix from the induced ensemble of parameters $(k,n)$ is $\operatorname{Dir}_k^{(n)}$, see \cite[Section VIII]{puchala2016distinguishability}.

Moreover, the probability range of a random POVM is related to the diagonals of the output set of a random quantum channel. In order to state and prove the main theorem, let us recall the definition of the \emph{$(t)$-norm} from \cite{belinschi2012eigenvectors}. To any vector $x \in \mathbb R^k$ associate a self-adjoint element in the non-commutative probability space $(\mathbb C^k, \operatorname{tr})$, where we denote by $\operatorname{tr} := \frac 1 k \tr{\cdot}$ the normalized trace. Consider also the projection $p$ of trace $t \in (0,1)$ living in the non-commutative probability space $(\mathbb C^2, \operatorname{tr})$. We define the $(t)$-norm of $x$ as 
$$\|x\|_{(t)} := \|p x p\|,$$
where the elements in the right hand side live in the \emph{free product} of the two non-commutative probability spaces mentioned above. Moreover, let us define the set 
$$K_{k,t} := \{\lambda \in \Delta_k \, : \, \forall a \in \Delta_k, \, \langle \lambda, a \rangle \leq \|a\|_{(t)}\}.$$

\begin{theorem}\label{thm:prob-ran-random-POVM}
Consider a sequence $(M^{(n)})_n$ of $k$-valued random POVMs, with effects $M^{(n)}_i \in \mathcal M_{d_n}(\mathbb C)$. The effect dimensions scale as $d_n \sim tkn$, for some constant $t \in (0,1)$. Almost surely, the probability ranges of the random POVMs $M^{(n)}$ converge to the deterministic convex set $K_{k,t}$, in the following sense:
$$K_{k,t}^\circ \subseteq \liminf_{n \to \infty} \operatorname{ProbRan}(M^{(n)}) \subseteq  \limsup_{n \to \infty} \operatorname{ProbRan}(M^{(n)}) \subseteq K_{k,t}.$$
\end{theorem}
\begin{proof}
The result for the output sets of the random quantum channels $\Psi^{(n)}$ is \cite[Theorem 6.2]{collins2015convergence}, which in turn builds upon \cite[Theorem 5.4]{belinschi2012eigenvectors}. Restricting to diagonals obviously preserves the upper bound, by the Schur-Horn theorem: for any Hermitian matrix $A$, $\operatorname{diag}(A) \prec \operatorname{spec}(A)$, where $\prec$ denotes the majorization relation, see \cite[Exercise II.1.12]{bhatia1997matrix}. For the lower bound, note that \cite[Theorem 6.2]{collins2015convergence} is stated at the level of matrices: any self-adjoint matrix having its spectrum in the interior of $K_{k,t}$ will eventually be in the output set of $\Psi^{(n)}$; in particular, this holds for diagonal matrices.
\end{proof}

In general, computing the $(t)$-norm of vectors in $\mathbb R^k$ requires solving polynomial equations of high degree. The only analytical result in a closed form is the value of the $(t)$-norm for bi-valued vectors. First, note that for non-negative reals $0 \leq a \leq b$, we have
$$\|(a, a, \ldots, a, b, b, \ldots, b\|_{(t)} = a + \|(0, 0, \ldots, 0, b-a, b-a, \ldots, b-a\|_{(t)}.$$
Then, if follows from \cite[Proposition 3.6]{belinschi2012eigenvectors} that
$$\|(\underbrace{1,1,\ldots, 1}_{j \text{ times}}, \underbrace{0,0,\ldots, 0}_{k-j \text{ times}})\|_{(t)}=
\begin{cases}
t +u-2tu +2\sqrt{tu (1-t)(1-u)}  &\text{ if } t + u < 1,\\
1  &\text{ if } t + u \geq 1,
\end{cases}$$
where $u = j/k \in [0,1]$. We have thus a complete picture of the asymptotic probability range for $k=2$: 
$K_{2,t} = \{(p,1-p) \, : \, |p-1/2|\leq x_t\}$, with 
$$x_t = \begin{cases}
\sqrt{t(1-t)}&\text{ if } t \leq 1/2\\
1/2&\text{ if } t > 1/2.
\end{cases}$$

%%%%%%%%%%%%%%%%%%%%%%%%%%%%
\section{(In-)Compatibility criteria for random POVMs}
%%%%%%%%%%%%%%%%%%%%%%%%%%%%
\label{Sec:comp-criteria}
%%%%%%%%%%%%%%%%%%%%%%%

Having developed in Sections \ref{Secc:random-POVM} and \ref{sec:stat-prop} the theory of random POVMs, we turn in this section to the question of compatibility of generic POVMs. The fundamental question here is the following:

\medskip

\emph{Given two independent random POVMs, what is the probability that they are compatible? }

\medskip

\noindent More precisely, the two random POVMs are chosen independently from the Haar ensembles with parameters $(d_i,k_i;n_i)$ respectively ($i=1,2$); we assume obviously that $d_1=d_2$. Since compatibility of random POVMs can be formulated as a semidefinite program, the considerations in this section could also be seen as giving bounds for the existence of solutions of random SDPs (see, e.g.~\cite{amelunxen2015intrinsic}).

As it is often the case in Random Matrix Theory, we shall focus on the asymptotic regime where the Hilbert space size $d=d_1=d_2$ grows to infinity. We shall keep the number of effects $k_{1,2}$ in the POVMs constant, and the respective parameters $n_{1,2}$ will follow linear scalings with respect to $d$; this is precisely the asymptotical regime studied in Proposition \ref{prop:limit-eigenvalues-Mi}. 

The first three subsections deal with \emph{compatibility criteria}, that is sufficient conditions for compatibility. The following two subsections are focused on \emph{incompatibility criteria}, i.e.~necessary conditions for compatibility; it turns out that the two such criteria we discuss are not informative in the asymptotical regime we investigate. Finally, we compare the noise content and the Jordan product criteria in the last subsection.

\subsection{The noise content criterion}
%%%%%%%%%%%%%%%%%%%%%%%

We analyze in this section the \emph{noise content criterion}, stated in Prop. \ref{prop:nc-criterion}, when applied to Haar-random POVMs. 

We know from Proposition \ref{prop:limit-eigenvalues-Mi} that, for a Haar-random POVM $A$ with parameters $(d,k;n)$, in the asymptotic regime where $k$ is fixed and $d,n \to \infty$ in such a way that $d \sim skn$ for some constant $s \in (0,1]$, the smallest eigenvalue of some POVM element $A_i$ converges, almost surely, to the constant $\varphi_-$ from \eqref{eq:phi-pm}
$$\varphi_-(s,k)= \begin{cases}
s +   \frac{1-2s}{k} - \frac{2}{k} \sqrt{s(1-s)(k-1)}&\qquad \text{ if } s < \frac 1 k\\
0 &\qquad \text{ if } s \geq \frac 1 k.
\end{cases}$$
The formula above allows us to obtain the limiting noise content of random POVMs: for a sequence of random POVMs of parameters $(d,k;n_d)$, where $n_d$ is a sequence of integers with the property that $d \sim skn_d$ (as $d \to \infty$), the noise content $w(M) = \sum_{i=1}^k \lambda_{\min}(M_i)$ converges, almost surely as $d \to \infty$ and $k, s$ fixed, to the quantity $k \varphi_-(s,k)$

Using this result, we obtain the following compatibility criterion for Haar-random POVMs. 

\begin{theorem}\label{thm:min-eig-criterion}
	Let $(A^{(d)})$, $(B^{(d)})$ be two sequences of random POVMs of respective parameters $(d, k;n_d)$ and $(d, l;m_d)$ where $n_d$ and $m_d$ are two integer sequences growing to infinity in such a way that $d \sim skn_d$ and $d \sim tlm_d$ for two constants $s,t\in (0,1]$. If 
	\begin{equation}\label{eq:limit-criterion-lambda-min}
		k\varphi_-(s,k) + l \varphi_-(t,l) >1,
	\end{equation}
	then, almost surely as $d \to \infty$, the Haar-random POVMs $A^{(d)}$ and $B^{(d)}$ are asymptotically compatible. 
\end{theorem}
\begin{proof}
	From Proposition \ref{prop:limit-eigenvalues-Mi}, we know that for individual POVM operators $A^{(d)}_i$ (resp.~$B^{(d)}_j$), the minimum eigenvalue converges, almost surely as $d \to \infty$, to the corresponding value $\varphi_-$ (here, we need the strong convergence flavor of the theorem). Taking the intersection of $k+l$ almost sure events, we obtain the simultaneous almost sure convergence of the sum of minimum eigenvalues to the left-hand-side of \eqref{eq:limit-criterion-lambda-min}. The conclusion follows from a standard countable approximation argument. 
\end{proof}

\begin{remark}
	Note that in Theorem \ref{thm:min-eig-criterion} we do not need to make any assumptions on the \emph{joint distribution} of the random POVMs $A$ and $B$ (such as independence). This is due to the fact that the minimum eigenvalue compatibility criterion we are using only depends on individual spectral characteristics of the two POVMs.
\end{remark}

\begin{corollary}\label{cor:min-eig-k-s}
	In the case $k=l\geq 2$ and $s=t$ (identically distributed Haar-random POVMs), the condition from \eqref{eq:limit-criterion-lambda-min} simplifies to 
	$$s < \frac{1}{6 k - 4 + 4 \sqrt{(k-1) (2k-1)}}.$$
\end{corollary}

\begin{corollary}\label{cor:min-eig-s-t}
	In the case $k=l=2$ and $s,t$ arbitrary (dichotomic POVMs), the condition from \eqref{eq:limit-criterion-lambda-min} simplifies to 
	$$t < \frac 1 2 - \sqrt{\sqrt{s(1-s)}-s(1-s)} \quad \text{ and } \quad s<\frac 1 2.$$
\end{corollary}

The condition $s<\frac 1 2$ appears because one needs to take the first branch of the definition of the function $\varphi_-$ in order to satisfy \eqref{eq:limit-criterion-lambda-min}. The inequalities of Cor.~\ref{cor:min-eig-k-s} and Cor.~\ref{cor:min-eig-s-t} are depicted in Fig.~\ref{fig:nc-vs-jordan}.

%%%%%%%%%%%%%%%%%%%%%%
\subsection{The Jordan product criterion}\label{sec:Jordan-product}
%%%%%%%%%%%%%%%%%%%%%%

In this section we focus on the compatibility criterion given by the \emph{Jordan product}, see Prop. \ref{prop:jordan-product-criterion}. 
To apply this criterion to Haar-random POVMs $A$ and $B$, one has to compute the minimum eigenvalue of the Jordan product $A_i \circ B_j$ of two (independent) random matrices having limiting eigenvalue distributions such as in Proposition \ref{prop:limit-eigenvalues-Mi}. The computation of the distribution of the anti-commutator of two random matrices is an important problem in the general theory of random matrices, which has received some attention in the last years, especially in the framework of free probability \cite{nica1998commutators,vasilchuk2003asymptotic}. A nice description of the anti-commutator of a pair of free random variables remains elusive in the most general case, despite some partial results (e.g.~for even random variables, see \cite[Proposition 1.10]{nica1998commutators}) and some implicit characterizations (see \cite[Theorem 2.2]{vasilchuk2003asymptotic}). 

In the absence of an analytical description of the smallest eigenvalue of the Jordan product of two random POVM elements, we rely here on the following general lower bound. For a positive definite matrix $X$, we denote
$$R(X):= \frac{\lambda_{\max}(X)}{\lambda_{\min}(X)} \in [1, \infty).$$

\begin{lemma}[\cite{strang1962eigenvalues, nicholson1979eigenvalue, alikakos1984estimates}]\label{lem:min-eig-Jordan-product}
	Let $X,Y \in \mathcal M_d(\mathbb C)$ be two positive definite matrices. If any of the two equivalent conditions below holds
	\begin{itemize}
		\item $(\sqrt{R(X)}-1)(\sqrt{R(Y)}-1) < 2$
		\item $(R(X)-1)^2(R(Y)-1)^2 < 16 R(X) R(Y)$
	\end{itemize}
	then $Z = X \circ Y = XY + YX$ is positive definite. 
\end{lemma}

The next result uses the previous lemma to provide a sufficient criterion for the asymptotic compatibility of Haar-random POVMs. We omit the proof, since it is very similar to the proof of Theorem \ref{thm:min-eig-criterion}. We need the following notation ($k \geq 2$ and $0 < s \leq 1$):
$$R(k,s) := \begin{cases}
\displaystyle{\frac{s + \frac{1-2s}{k} + \frac{2}{k} \sqrt{s(1-s)(k-1)}}{s + \frac{1-2s}{k} -\frac{2}{k} \sqrt{s(1-s)(k-1)}}}&\qquad \text{ if } s < \frac 1 k\\
+\infty  &\qquad \text{ if } s \geq \frac 1 k.
\end{cases}$$

\begin{theorem}\label{thm:jordan-product-criterion}
	Let $(A^{(d)})$, $(B^{(d)})$ be two sequences of random POVMs of respective parameters $(d, k;n_d)$ and $(d, l;m_d)$ where $n_d$ and $m_d$ are two integer sequences growing to infinity in such a way that $d \sim skn_d$ and $d \sim tlm_d$ for two constants $s,t\in (0,1]$. If 
	\begin{equation}\label{eq:limit-criterion-jordan}
	\left(\sqrt{R(k,s)}-1\right)\left(\sqrt{R(l,t)}-1\right) < 2,
	\end{equation}
	then, almost surely as $d \to \infty$, the Haar-random POVMs $A^{(d)}$ and $B^{(d)}$ are asymptotically compatible. 
\end{theorem}

\begin{remark}
 	As for Theorem \ref{thm:min-eig-criterion}, we do not need to make any assumptions on the joint distribution of the random POVMs $A$ and $B$. Although the Jordan product compatibility criterion depends jointly and in a non-trivial manner on the POVM elements of both $A$ and $B$, the inequality from Lemma \ref{lem:min-eig-Jordan-product} separates these contributions, allowing for the very general bound \eqref{eq:limit-criterion-jordan}.
\end{remark}

\begin{corollary}\label{cor:jordan-k-s}
	In the case $k=l\geq 2$ and $s=t$ (identically distributed Haar-random POVMs), the condition from \eqref{eq:limit-criterion-jordan} simplifies to $R(k,s)<3+2\sqrt 2$, which, after some algebra, yields
	$$s < \frac{ k(3-2 \sqrt 2) + 2(\sqrt 2 -1)}{k^2+4k-4} < \frac 1 k.$$
\end{corollary}

\begin{corollary}\label{cor:jordan-s-t}
	In the case $k=l=2$ and $s,t$ arbitrary (dichotomic POVMs), the condition from \eqref{eq:limit-criterion-jordan} simplifies to 
	$$\sqrt{s(1-s)} + \sqrt{t(1-t)} < \frac 1 4 \iff t < \frac 1 2 - \sqrt{s(1-s)} \quad \text{ and } \quad s<\frac 1 2.$$
\end{corollary}

The inequalities of Cor.~\ref{cor:jordan-k-s} and Cor.~\ref{cor:jordan-s-t} are depicted in Fig.~\ref{fig:nc-vs-jordan}.

%%%%%%%%%%%%%%%%%%%%%%
\subsection{The optimal cloning map criterion}
%%%%%%%%%%%%%%%%%%%%%%

We briefly discuss here the optimal cloning compatibility criterion presented in Proposition \ref{prop:optimal-cloning} for Haar-random POVMs. The relevant quantities here are the minimal eigenvalues of the effects (which were discussed at length in Proposition \ref{prop:limit-eigenvalues-Mi} and used in Theorem \ref{thm:min-eig-criterion}) and the traces of the effects. Regarding the latter quantities, we know from Proposition \ref{prop:limit-eigenvalues-Mi} that, almost surely as $d \to \infty$,
$$\forall i\, : \, \lim \frac{\tr{A^{(d)}_i}}{d} = \frac 1 k$$
for a sequence of random Haar-POVMs $A^{(d)}$ with parameters $(d, k;n_d)$ in the scaling $d \sim skn_d$. It follows that the asymptotical version of equation \eqref{eq:cloning-A} reads
$$ \varphi_-(s,k) > \frac{1}{2k}.$$
Assuming also the corresponding condition for a second sequence of random Haar-POVMs $B^{(d)}$ with parameters $(d, l;m_d)$, we recover by summing them equation \eqref{eq:limit-criterion-lambda-min}, showing that, asymptotically, for Haar-random POVMs, the optimal cloning criterion is weaker that the noise content criterion. Note however that this is not the case at fixed dimension $d$, as it was pointed out in Remark \ref{rem:cloning-vs-noise-content-fixed-d}.

%%%%%%%%%%%%%%%%%%%%%%%%%%%%
\subsection{Unsharpness and the Miyadera-Imai criterion}
%%%%%%%%%%%%%%%%%%%%%%%%%%%%

Our goal in this section is to analyze under which conditions independent random POVMs are certified incompatible by the Miyadera-Imai criterion recalled in Section \ref{sec:MI-criterion}.

Since the sharpness measure from Definition \ref{def:sharpness} plays an important role in the Miyadera-Imai criterion, let us study it in the case of the random POVMs defined in Section \ref{Secc:random-POVM}.

\begin{proposition}\label{prop:random-sharpness}
	Let $(A^{(d)})$ be a sequence of random POVMs of parameters $(d, k;n_d)$ where $n_d$ is an integer sequence growing to infinity in such a way that $d \sim skn_d$ for a constant $s\in (0,1]$. Then, almost surely, for all $i=1,\ldots, k$,
	\begin{equation}\label{eq:limit-unsharpness}
	\lim_{d \to \infty} \sigma(M_i^{(d)}) = \sigma(k,s):=
\begin{cases}
4\varphi_+(s,k^{-1})(1-\varphi_+(s,k^{-1})), &\qquad \text{ if } s \in [0,s_0)\\
1, &\qquad \text{ if } s\in [s_0, 1-s_0]\\
4\varphi_-(s,k^{-1})(1-\varphi_-(s,k^{-1})), &\qquad \text{ if } s \in (1-s_0,1)\\
0, &\qquad \text{ if } s=1.
\end{cases}
	\end{equation}
	where $\sigma(\cdot)$ is the sharpness measure from \eqref{eq:def-sharpness}, $\varphi_\pm$ are the constants defined in \eqref{eq:phi-pm}, and 
$$s_0 = \frac 1 2 - \frac{\sqrt{k-1}}{k} \in [0,1/2).$$
\end{proposition}
\begin{proof}
The result follows from a simple analysis of the support of the measure \eqref{eq:limit-measure-Mi}.
\end{proof}
\begin{remark}
It is easy to see that the limiting value $\sigma(k,s)$ is symmetric w.r.t.~$s=1/2$: $\sigma(k,s) = \sigma(k,1-s)$. Moreover, for all $k$ and $s \in (0,1)$, $\sigma(k,s) \geq k^{-1}-k^{-2}$. At $s=1$, the random POVM elements $M_1, \ldots, M_k$ are random projections summing to the identity, hence the unsharpness is null.
\end{remark}
\begin{remark}
For $k=2$, we have $s_0=0$ and thus, for all $s \in [0,1]$, $\sigma(2,s) = 1$. This is because $1/2$ is, asymptotically, almost surely an element of the spectrum of both effects of a binary random POVM.
\end{remark}

In Figure \ref{fig:unsharpness}, we plot the limiting value of the unsharpness $\sigma(A_i)$ as a function of $s$, for fixed $k$. 

\begin{figure}[htbp]
	\begin{center}
		\includegraphics[scale=0.60]{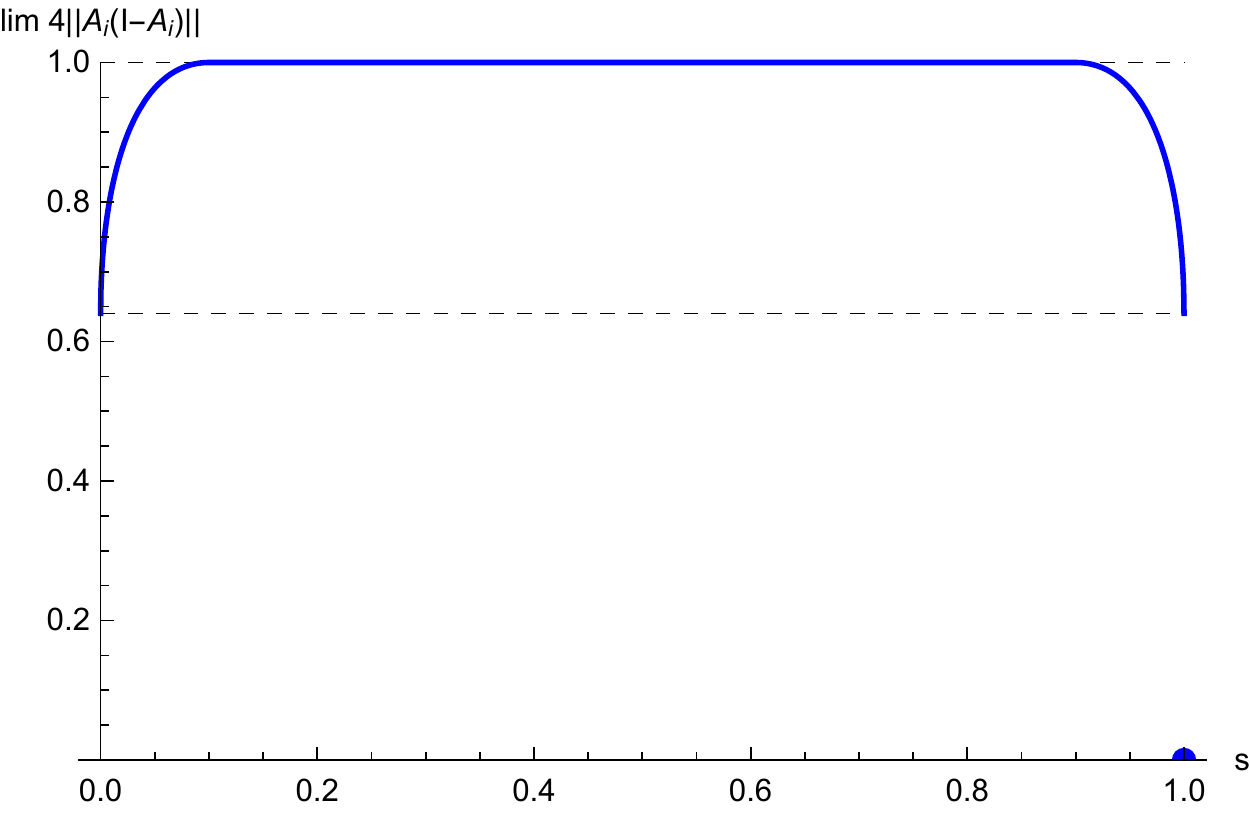} \qquad  \includegraphics[scale=0.60]{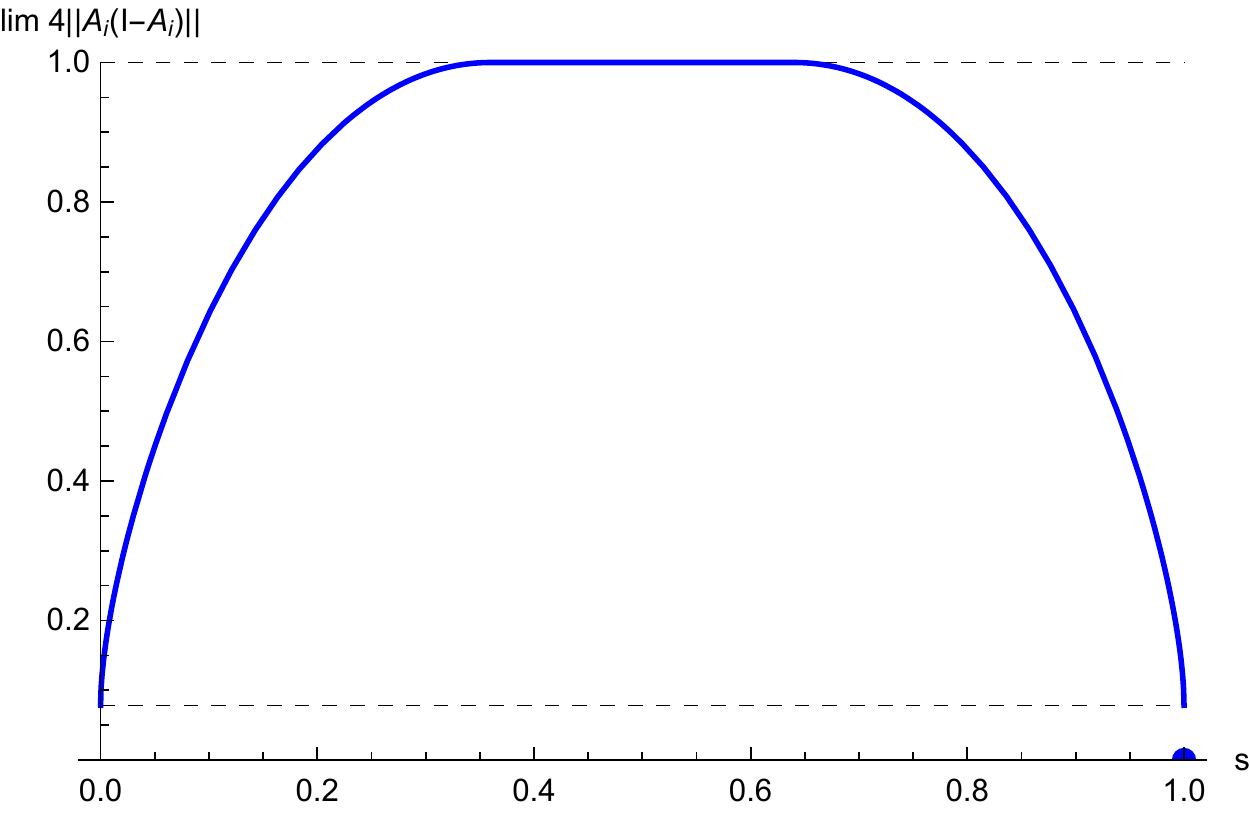}	
		\caption{The limiting value of the unsharpness of a random POVM as a function of $s$, for $k=5$ (left) and $k=50$ (right). The dotted horizontal lines correspond to the minimal ($4(k^{-1}-k^{-2})$) and maximal (1) values of the unsharpness, corresponding to a fixed value of $k$.}
		\label{fig:unsharpness}
	\end{center}
\end{figure}

Regarding now the application of the Miyadera-Imai criterion, since this is only a necessary condition for compatibility of POVMs, the only scenario in which it can be used is if
\begin{equation}\label{eq:MI-violation}
4\|[A_i, B_j]\|^2 > \sigma(A_i) \sigma(B_j),
\end{equation}
in which case the POVMs are guaranteed to be incompatible. Above, $A_i$ and $B_j$ are quantum effects belonging to two POVMs $A$ and $B$. In the case of random POVMs $A$ and $B$, the difficulty lies in computing the left hand side. For example, in the most natural setting, when $A^{(d)}$ and $B^{(d)}$ are sequences of independent, identically distributed Haar-random POVMs as in Proposition \ref{prop:random-sharpness}, one needs to compute the limiting eigenvalue distribution of the random matrix $A_iB_j - B_jA_i$ in a strong sense (in order to also obtain the convergence of the operator norm). As it was argued in Section \ref{sec:Jordan-product}, the computations of the limiting distributions of commutators and anti-commutators is a highly non-trivial question in free probability,  so we lack a precise answer in our setting. More recent theoretical results based on the theory of operator-valued free probability \cite{belinschi2017analytic} might be the right framework to tackle such questions; we leave this question open. Numerical simulations\footnote{Numerical routines for this paper can be found in the supplementary material of the arXiv version.} seem to suggest however that the Miyadera-Imai criterion does not allow to conclude that Haar-random POVMs are incompatible, see Table \ref{tab:MI-criterion}. 

\begin{table}
\begin{tabular}{|r||c|c|c|c|c|c|c|c|c|}
	\hline
k\textbackslash s	& 0.1 & 0.3 & 0.5 & 0.7 & 0.9\\
\hline\hline
2&0.025522&0.19976&0.47205&0.75635&0.963\\\hline
3&0.020689&0.17055&0.42674&0.72226&0.95712\\\hline
5&0.011771&0.11201&0.31746&0.60386&0.91434\\\hline
\end{tabular}
\medskip
\caption{The average value of $4\|[A_i, B_j]\|^2$ for 10 pairs of independent quantum effects $A_i, B_j$ from the Haar-random POVM ensemble of parameters $(d = \lfloor skn \rfloor, k; n=1000)$ for $s=0.1, 0.3, 0.5, 0.7, 0.9$ and $k=2,3,5$. In all these cases, the right hand side of \eqref{eq:MI-violation} is equal to 1, asymptotically, so the application of the Miyadera-Imai criterion is inconclusive.}
\label{tab:MI-criterion}
\end{table}

%%%%%%%%%%%%%%%%%%%%%%
\subsection{The Zhu criterion}
%%%%%%%%%%%%%%%%%%%%%%

The Zhu criterion from Proposition \ref{prop:Zhu} is unfortunately uninformative in the asymptotical regime we are interested in. Indeed, using the triangle inequality, we upper bound the expression of $\tau$ from \eqref{eq:tau-Zhu-1-norm} by 
$$\frac{1}{2}\big[ \tr{\mathcal{G}_A}+\tr{\mathcal{G}_B}+\|\mathcal{G}_A\|_1+\|\mathcal{G}_B\|_1\big] = \tr{\mathcal{G}_A}+\tr{\mathcal{G}_B}.$$
Since $\tr{\mathcal{G}_A} = \sum_{i=1}^k \tr{A_i^2} / \tr{A_i} \leq k$, to obtain a violation of the inequality from Proposition \ref{prop:Zhu}, one needs $k+l > d$, where $k$ and resp.~$l$ is the number of outcomes of the POVMs $A$, resp.~$B$. In the regime we are interested in (fixed number of outcomes, large dimension), this inequality can not hold, so the Zhu criterion is not applicable in our setting. 

\subsection{Comparing the different compatibility criteria}
\label{Sec:comp}

We compare in this section the different compatibility criteria described previously. As it was noted, the optimal cloning criterion is asymptotically weaker than the noise content criterion, so we do not discuss it here. The two relevant criteria are the Jordan product criterion and the noise content criterion. We compare them in Figure \ref{fig:nc-vs-jordan}, and we notice that the Jordan product criterion performs systematically better. This result is surprising, since we have used several inequalities in our analysis from Section \ref{sec:Jordan-product} in order to be able to apply the criterion to random matrices.

We conclude that, in the presence of typical random POVMs, one has interest in checking first the Jordan product criterion in order to certify compatibility. 

\begin{figure}[htbp]
	\begin{center}
		\includegraphics[scale=0.60]{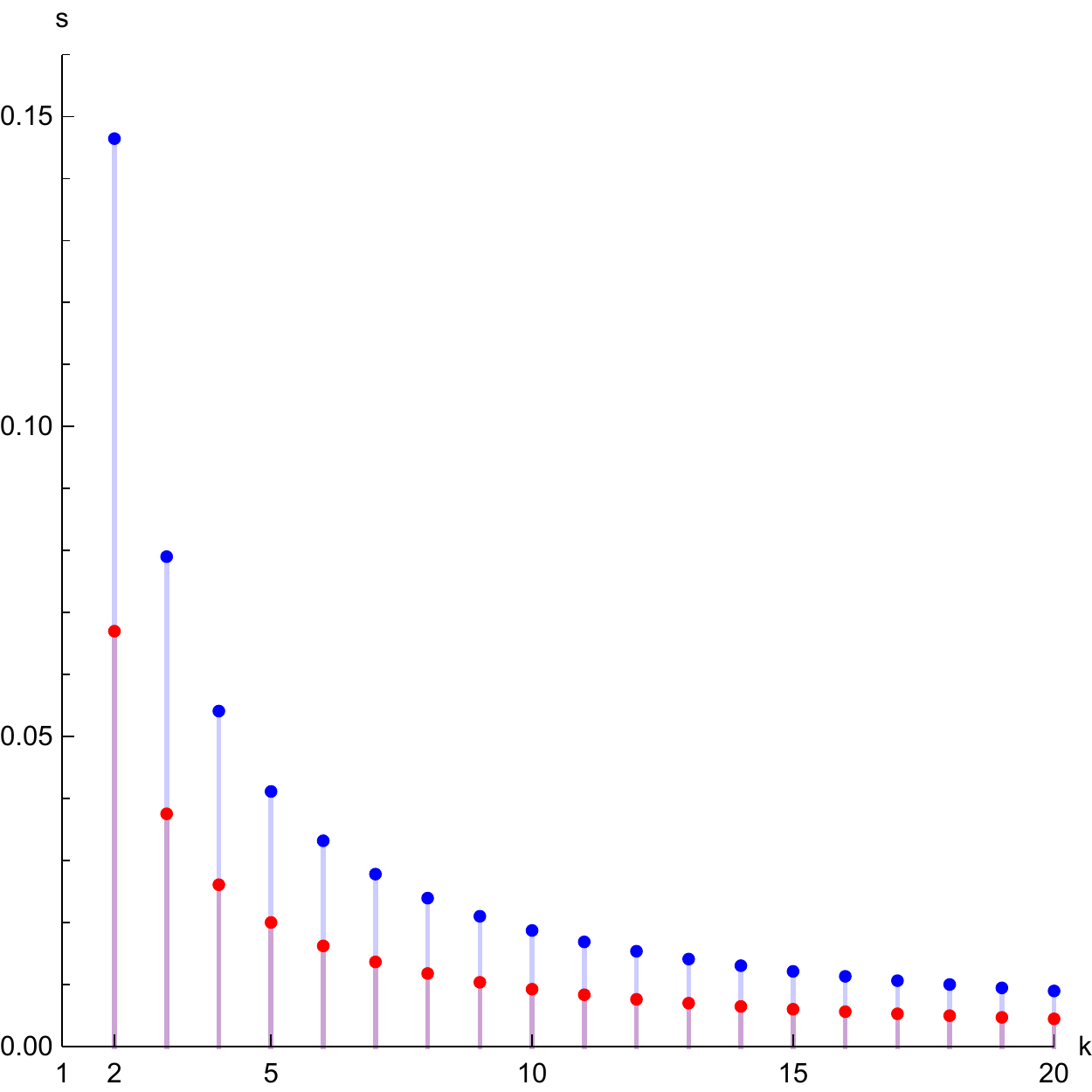} \qquad  \includegraphics[scale=0.60]{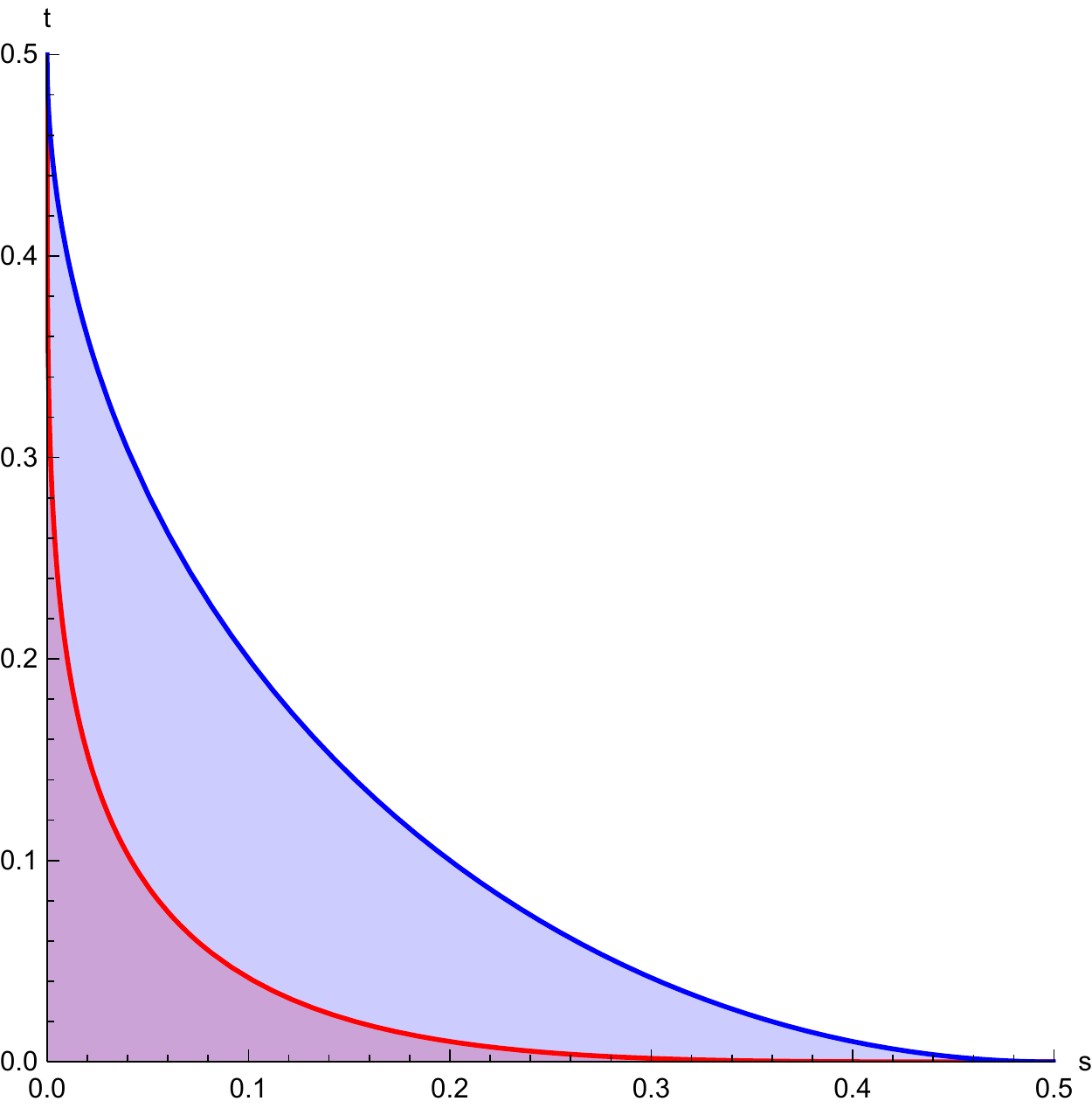}	
		\caption{Efficiency of the noise content and Jordan product criteria. Left: the range of values $(k,s)$ for which one can infer the compatibility of two identically distributed Haar-random POVMs (see Corollaries \ref{cor:min-eig-k-s} and \ref{cor:jordan-k-s}). Right: the range of values $(s,t)$ for which one can infer the compatibility of two dichotomic Haar-random POVMs (see Corollaries \ref{cor:min-eig-s-t} and \ref{cor:jordan-s-t}). The top curves (in blue) correspond to the Jordan product criterion, while the bottom ones (in red) to the noise content criterion.}
		\label{fig:nc-vs-jordan}
	\end{center}
\end{figure}

We would like also to point out that, at this time, we do not know of any incompatibility criteria which would give any insightful information in the asymptotic regime studied in this paper (fixed number of outcomes, large matrix dimension). It would be interesting to develop such criteria, adapted to noisy POVMs. 

\appendix
\section{Density of Wishart-random POVMs}

We prove in this appendix Theorem \ref{thm:density-Wishart-POVM} with the help of the matrix Dirac delta functions \cite{hoskins2009delta,zhang2016dirac}. Let us recall that a Wishart-random POVM $M$ of parameters $(d,k;s_1, \ldots, s_k)$ is obtained by normalizing $k$ independent Wishart matrices $W_1, \ldots, W_k$ of respective parameters $(d,s_i)$:
$$M_ i = S^{-1/2} W_i S^{-1/2}, \qquad \text{ where } S = \sum_{j=1}^k W_j.$$
We also recall that a Wishart random matrix of parameters $(d,s)$ has density
$$\frac{\mathrm{d} \mathbb P}{\mathrm{d} \mathrm{Leb}}(w) = C_{d,s} \mathbf{1}_{w \geq 0} \exp(-\operatorname{Tr} w) (\det w)^{s-d}.$$

\begin{proof}[Proof of Theorem \ref{thm:density-Wishart-POVM}]
In the course of the proof, we shall not keep track of constants, although this could be done with the help of the Weyl integration formula \cite[Proposition 4.1.3]{anderson2010introduction} and the Selberg integral \cite[Eq.~(17.6.5)]{mehta2004random}. We have
\begin{align*}
\frac{\mathrm{d} \mathbb P}{\mathrm{d} \mathrm{Leb}}&(m_1, \ldots, m_k) \sim \int \prod_{i=1}^k \mathrm{d}W_i \mathbf{1}_{W_i \geq 0}\exp(-\operatorname{Tr} W_i) (\det W_i)^{s_i-d} \\
&\qquad \qquad \qquad \qquad \qquad \cdot\delta\big[m_i - \big(\sum_j W_j\big)^{-1/2}W_i \big(\sum_j W_j\big)^{-1/2}\big]\\
&= \int \mathrm{d}S \big(\prod_{i=1}^k \mathrm{d}W_i\big)\delta(S-\sum_j W_j)\prod_{i=1}^k  \mathbf{1}_{W_i \geq 0}\exp(-\operatorname{Tr} W_i) (\det W_i)^{s_i-d} \delta(m_i - S^{-1/2}W_i S^{-1/2}).
\end{align*}
We shall now make a change of variables $W_i = S^{1/2}m_i^{1/2}Y_im_i^{1/2}S^{1/2}$, where $S$ and $m_i$ are treated like constants and $Y_i$ are the new variables. Computing the Jacobian of this transformation (see also \cite[Proposition 3.7]{zhang2015volumes}), we have
$$\mathrm{d}W_i = (\det S)^d (\det m_i)^d \mathrm{d}Y_i.$$
Factorizing the expressions appearing in the delta functions, and using \cite[Proposition 3.3]{zhang2016dirac}, we get
\begin{align*}
\delta\big[m_i - \big(\sum_j W_j\big)^{-1/2}W_i \big(\sum_j W_j\big)^{-1/2}\big] &= (\det m_i)^{-d} \delta(I_d - Y_i)\\ 
\delta(S-\sum_j W_j) &= (\det S)^{-d} \delta(I_d - \sum_j m_j^{1/2} Y_j m_j^{1/2}).
\end{align*}
Plugging everything into the expression for the density, we obtain
\begin{align*}
\frac{\mathrm{d} \mathbb P}{\mathrm{d} \mathrm{Leb}}(m_1, \ldots, m_k) &\sim \int \mathrm{d}S\big(\prod_{i=1}^k \mathrm{d}Y_i\big) \mathbf{1}_{S \geq 0}\exp(-\operatorname{Tr} S)(\det S)^{(k-1)d} \delta(I_d - \sum_j m_j^{1/2} Y_jm_j^{1/2})\\
&\qquad\qquad\qquad\qquad\cdot\prod_{i=1}^k \mathbf{1}_{Y_i \geq 0} (\det S)^{s_i-d}(\det m_i)^{s_i-d}(\det Y_i)^{s_i-d} \delta(I_d - Y_i)\\
&=\delta(I_d - \sum_j m_j) \prod_{i=1}^k (\det m_i)^{s_i-d} \int \mathrm{d}S \mathbf{1}_{S \geq 0}\exp(-\operatorname{Tr} S)(\det S)^{\sum_j s_j-d}\\
&\sim\delta(I_d - \sum_j m_j) \prod_{i=1}^k (\det m_i)^{s_i-d} ,
\end{align*}
which is formula \eqref{eq:density-Wishart-POVM}, finishing the proof.
\end{proof}

\bibliographystyle{alpha}
\bibliography{compatibility}

\end{document}